\documentclass[12pt]{article}
\usepackage[utf8]{inputenc}

\usepackage[dvipdfmx]{graphicx}
\usepackage{url}
\usepackage{listings,jvlisting}
\usepackage{amssymb}
\usepackage{tabularx}
\usepackage{amsmath}
\usepackage{mdframed}
\usepackage{graphicx}
\usepackage{setspace}
\usepackage{here}
\usepackage{authblk}
\usepackage{mathtools}
\usepackage{amsthm}
\usepackage{algorithmic}
\usepackage{algorithm}
\usepackage{arydshln}
\usepackage{longtable}
\usepackage{setspace}
\usepackage{amsthm}
\usepackage[authoryear]{natbib}
\usepackage{subcaption}
\usepackage{bbm}

\newtheorem{theorem}{Theorem}

\newtheorem{example}{Example}
\newtheorem{lemma}{Lemma}
\newtheorem{prop}{Proposition}

\newcommand{\indep}{\mathop{\perp\!\!\!\perp}}

\def\ep{{\varepsilon}}

\def\bal{{\text{\boldmath $\alpha$}}}
\def\bbe{{\text{\boldmath $\beta$}}}

\def\bep{{\text{\boldmath $\epsilon$}}}

\def\bsi{{\text{\boldmath $\sigma$}}}

\def\bmu{{\text{\boldmath $\mu$}}}

\def\Si{{\Sigma}}

\def\bSi{{\text{\boldmath $\Si$}}}

\def\e{{\text{\boldmath $e$}}}

\def\x{{\text{\boldmath $x$}}}

\def\z{{\text{\boldmath $z$}}}

\def\I{{\text{\boldmath $I$}}}

\def\Q{{\text{\boldmath $Q$}}}

\def\W{{\text{\boldmath $W$}}}
\def\X{{\text{\boldmath $X$}}}
\def\Y{{\text{\boldmath $Y$}}}

\def\diag{{\rm diag\,}}

\def\cS{{\mathcal{S}}}
\def\cP{{\mathcal{P}}}
\def\cM{{\mathcal{M}}}

\begin{document}
\title{Integrate Meta-analysis into Specific Study (InMASS) for Estimating Conditional Average Treatment Effect}
\author[1]{Keisuke Hanada \footnote{Address: Graduate School of Engineering Science, Osaka University, 1-3 Machikaneyama-cho, Toyonaka, Osaka, Japan. \quad
E-Mail: hanada.keisuke.es@osaka-u.ac.jp}}
\author[2]{Masahiro Kojima}
\affil[1]{Osaka University}
\affil[2]{The Institute of Statistical Mathematics}
\maketitle

\abstract{\noindent
Randomized controlled trials are the standard method for estimating causal effects, ensuring sufficient statistical power and confidence through adequate sample sizes. However, achieving such sample sizes is often challenging. This study proposes a novel method for estimating the average treatment effect (ATE) in a target population by integrating and reconstructing information from previous trials using only summary statistics of outcomes and covariates through meta-analysis. The proposed approach combines meta-analysis, transfer learning, and weighted regression. Unlike existing methods that estimate the ATE based on the distribution of source trials, our method directly estimates the ATE for the target population. The proposed method requires only the means and variances of outcomes and covariates from the source trials and is theoretically valid under the covariate shift assumption, regardless of the covariate distribution in the source trials. Simulations and real-data analyses demonstrate that the proposed method yields a consistent estimator and achieves higher statistical power than the estimator derived solely from the target trial.
}
\par\vspace{4mm}

\section{Introduction}
In clinical trials and observational studies, outcomes and covariates are carefully considered to compare treatment and control groups within the target population. With the widespread adoption of evidence-based medicine (EBM), clinical trials and observational studies are required to produce scientifically valid results. Sufficient sample sizes are necessary to draw reliable conclusions from group comparisons. However, increasing sample sizes often requires a long time for participant recruitment and evaluation. In cases of rare diseases or limited target populations, achieving the required sample size is challenging, and studies may even be discontinued due to sample size limitations \citep{rees2019noncompletion}. Against this background, ``data borrowing''---which integrates results from previous individual clinical trials or meta-analyses as historical controls or as part of the treatment group in ongoing clinical trials or observational studies---has garnered increasing attention in recent years \citep{lewis2019borrowing, chu2021dynamic, kojima2023dynamic}. While research designs tailored to specific populations allow for clear interpretation of comparative results, such designs often become small sample sizes, which necessitate considerable time to achieve adequate enrollment \citep{sully2013reinvestigation}. 
In such cases, information from completed trials or meta-analyses is helpful to achieve sufficient statistical power. 
From a causal inference perspective, external data can serve as a valuable resource for estimating the conditional average treatment effect (CATE) for the target population \citep{goring2019characteristics, zhou2021incorporating}.

Meta-analyses and external trials have been incorporated into target trials and observational studies, but existing methods assume that the results of the target trial can be extrapolated to the overall population. 
For example, \cite{rover2020dynamically} proposed incorporating information from non-target trials using shrinkage estimation, which relies solely on the assumptions of the random-effects model. While this method can be applied within the meta-analytic framework using only aggregate data (AD) from the target trial, it implicitly assumes population homogeneity and cannot account for covariate information. 
The Matching-Adjusted Indirect Comparison (MAIC) method enables indirect comparisons between two treatments by matching individual participant data (IPD) from a target trial to the baseline AD of another trial \citep{signorovitch2012matching}. 
The Simulated Treatment Comparison (STC) approach constructs an outcome model from a target trial and compares the model with that of another trial \citep{caro2010no}, but STC is biased when covariate distributions differ significantly between trials. 
Parametric G-computation can reduce the bias occurring at STC by creating pseudo-populations \citep{remiro2022parametric}. 
In the Bayesian framework, the Meta-Analytic-Predictive (MAP) prior method can estimate target trial effects by incorporating meta-analytic results as a prior distribution to borrow historical information \citep{neuenschwander2010summarizing}. 
These approaches borrow external information from AD in previous trials and guarantee comparability by aligning the covariate distributions of the target trial with those of the source trials. 
However, the true interest is the treatment effect for the population of the target trial, not the population evaluated in previous trials. 
This difference may cause a gap between the desired and the practically estimated effects. We propose a novel inference framework that integrates meta-analysis into a specific study (InMASS) through transfer learning to “transfer” the AD from existing trials to the population of the target trial.

The scenarios for borrowing information from existing trials include: (i) using the information for part or all of the control group of the target trial, and (ii) integrating other real-world data (RWD) into an observational study of interest. The scenario (i) applies to cases where the required sample size calculated during trial planning cannot be fully collected for some reason, which necessitates the borrowing of additional participants from existing trials. 
For example, we consider a two-arm randomized controlled trial (RCT) requiring 100 participants per group. If collecting a total of 200 participants is impractical but collecting 120 participants is feasible, a trial with 60 participants per group may lack sufficient statistical power.
In such cases, borrowing information from existing trials can help achieve the required power. 
In clinical trials for new drug development, treatment group participants may not be available from previous trials, so only control group data may be borrowed. For instance, in a trial with 120 participants, 100 may be allocated to the treatment group and 20 to the control group, with additional information borrowed from existing trials. Scenario (ii) involves new observational studies aiming to enhance inference performance within the target population by incorporating results from previous observational studies. In such cases, information on both treatment and control groups is desired; however, previous studies often provide only AD. Thus, integrating AD using meta-analysis and applying the result of meta-analysis to the inference of the new observational study is a valuable approach. These settings are formalized with notation in Section \ref{sec-accessible-info}.

The major contributions of InMASS are summarized as follows. First, InMASS can estimate treatment effects for the target trial population, in contrast to existing borrowing approaches that estimate treatment effects for either the general population or the source trial population. Second, source trials only require the means and variances of outcomes and covariates, without necessitating IPD. InMASS guarantees that borrowing is feasible as long as an AD meta-analysis can be performed on the source trials, which implies that InMASS can be applied in general settings where the means and variances are stable. InMASS shares similarities with Federated Adaptive Causal Estimation (FACE) \citep{Han21012025} in that both aim to obtain the treatment effect for a target population. However, while FACE focuses on aggregating summary results across sites, our approach focuses on the integration of trials. This distinction is reflected in their assumptions: whereas FACE requires the presence of both treatment and control groups at all sites, InMASS allows for single-arm trials.

InMASS employs three tools: meta-regression, transfer learning, and weighted regression. Only AD is available for source trials; thus, direct matching between the IPD of source and target trials is not possible. The meta-regression identifies the relationship between outcomes and covariates in the source trials, thereby reconstructing their IPD. 
The reconstructed IPD aligns with the second moments of the outcomes and covariates in the source trials, which is sufficient for estimating the CATE, as explained in Proposition \ref{prop1}. Transfer learning applies the reconstructed IPD to the target trial under the assumption of covariate shift \citep{shimodaira2000improving}. The covariate shift guarantees comparability between the reconstructed IPD and the IPD in the target trial.
Under this covariate shift, weighted regression can estimate the CATE for the target trial using the density ratio of covariates between the target trial and the source trials. By combining these three tools, InMASS borrows information from the source trials and enhances statistical power for estimating the CATE in the target trial.

The remainder of the paper is organized as follows. In Section \ref{sec-prelim}, we introduce the study setting, notation, and data assumptions required for identifying the CATE. In Section \ref{sec-method}, we present InMASS and demonstrate its theoretical validity. We assess the theoretical properties through simulations and illustrate the application of InMASS with a case study in Section \ref{sec-num-result}. Finally, we conclude in Section \ref{sec-conclusion} with a summary of our findings and a discussion of future research directions.

\section{Preliminaries}\label{sec-prelim}

\subsection{Estimand: Conditional average treatment effect}  
Let $\mathcal{P}$ denote the overall population, which includes both the target and source study populations. Let $\X$ be a $p \times 1$ random vector of covariates, $Y$ a random variable representing the outcome, and $Z$ a random variable indicating group assignment, where $Z = 0$ corresponds to the control group and $Z = 1$ to the treatment group. For a subject $i \in \mathcal{P}$, the observed data $(\x, y, z)$ are drawn from the probability distribution $P_{\mathcal{P}}(\X, Y, Z)$.

In this study, we are interested in estimating the conditional average treatment effect (CATE) for the population of interest $T \subset \mathcal{P}$, which is a fundamental problem in causal inference \citep{holland1986statistics, vegetabile2021distinction}. We first define the average treatment effect (ATE) for the overall population $\mathcal{P}$ as
$$
\delta_{\mathcal{P}} = \mathrm{E}_{\mathcal{P}} \left[Y_{1} - Y_{0}\right] = \mathrm{E}_{\mathcal{P}} \left[Y \mid Z=1\right] - \mathrm{E}_{\mathcal{P}} \left[Y \mid Z=0\right],
$$
where $Y_{a}$ ($a=0,1$) represents the potential outcomes, $\mathrm{E}_{\mathcal{P}}[\, \cdot \,] = \int_{\mathcal{P}} \, \cdot \, dP_{\mathcal{P}}(x,y,z)$ denotes the expectation over the population $\mathcal{P}$, and $Y = ZY_{1} + (1-Z)Y_{0}$. The CATE for a specific subpopulation $S \subset \mathcal{P}$ is defined as
\begin{align}\label{cate}
    \delta_{S} = \mathrm{E}_S\left[Y_{1} - Y_{0}\right].
\end{align}
Our target estimand is $\delta_T$, which we estimate using the IPD of the target trial $T$ and AD from meta-analysis.

Let $S_k \subset \cP$ denote the population of the $k$-th study. We adopt the assumption given by \cite{rosenbaum1983central} for estimating the CATE. This assumption requires that the meta-analysis trials and the target trial include at least one randomized controlled trial.
Suppose that $\cS = \{T, S_1, \dots, S_K\} \subset \cP$ represents the complete set consisting of the target trial and the trials included in the meta-analysis, and let $\cM = \cS \setminus \{T\}$ denote the meta-analysis set from the overall population. We define the following assumption:
\begin{enumerate}
    \item[(A1)] \textbf{Conditional independence}: There exists at least one study $S \in \cS$ such that $0 < P(Z=1 \mid S) < 1$ and $\{Y_{1}, Y_{0}\} \indep Z \mid S$.
\end{enumerate}
Assumption (A1) requires that at least one trial includes both treatment and control groups and that the outcome is independent of the treatment assignment due to randomization or similar mechanisms. Thus, conditional independence is not required to hold across all populations. Since the set $\cS$ may include single-arm trials and non-randomized trials, (A1) is a weaker assumption than the standard version of conditional independence. From (A1), conditional positivity for the source population $\cS$, defined as $0 < P(Z = 1 \mid \cS) < 1$, holds. Then, provided that $0 < P(Z=1 \mid T) < 1$, we can identify the CATE in the target trial $T$ as 
\[
\delta_T = \mathrm{E}_T[Y_{1} - Y_{0}] = \mathrm{E}_T[Y \mid Z=1] - \mathrm{E}_T[Y \mid Z=0].
\]
However, the CATE cannot be estimated when the target trial is single-armed. We aim to efficiently estimate $\delta_T$ by combining both the target and source trials $\cS$ when direct estimation of $\delta_T$ is infeasible or challenging.

\subsection{Assumption of accessible information}
\label{sec-accessible-info}
A standard approach for estimating $\delta_T$ is to summarize IPD solely from the target population $T$. 
Let $X_S$ be the covariate matrix for the set $S$, and let $Y_{S,j}$ represent the outcome for group $j$ in $S$. 
The observed covariates and outcome for subject $i$ ($i=1, \dots, n_S$) in $S$ are denoted by $\x_{Si}$ and $y_{Si}$, respectively. 
In this case, the CATE for population $T$, $\delta_T$, is estimated as 
\[
\hat{\delta}_{T} = \mathrm{E}[Y_{T,1}] - \mathrm{E}[Y_{T,0}] = \bar{y}_{T1} - \bar{y}_{T0},
\]
with the variance estimator given by 
\[
\hat{\sigma}_{T}^2 = n_{T1}^{-1} \hat{\sigma}_{T1}^2 + n_{T0}^{-1} \hat{\sigma}_{T0}^2.
\]
Here, $\bar{y}_{Tj}$ denotes the mean of the outcomes for group $j$ in $T$, and $\hat{\sigma}_{Tj}^2$ represents the estimated variance.

On the other hand, IPDs from completed trials are often inaccessible. Without IPDs, calculating propensity scores or similar balancing scores to match participants in the meta-analysis set $\cM$ with those in the target population $T$ becomes challenging. Most existing approaches estimate these scores by aligning participants in $T$ with the AD from $\cM$ \citep{rover2020dynamically, signorovitch2012matching, caro2010no, remiro2022parametric, neuenschwander2010summarizing}.
However, by aligning the IPD of $T$ with the AD of $\cM$, these methods estimate $\delta_{\cS}$ rather than $\delta_T$. For the $k$-th trial conducted on the population $S_k$, we can observe, for each group $j\in\{0,1\}$, the number of participants $n_{kj}$, the mean of the outcomes $\bar{y}_{kj}$, the mean of the covariates $\bar{\x}_{kj}$, the variance of the outcomes $\sigma_{kj}^2$, and the variance of the covariates $\bsi_{\bar{\x}_{kj}}^2 \I_p$, respectively. Here, $\I_p$ denotes the $p\times p$ identity matrix, and we assume that the information obtained from each trial is independent across trials. The covariates are not limited to continuous variables; categorical variables representing frequencies or proportions are also acceptable. 
For instance, we consider a binary covariate $X_{bi}$, $i=1,\dots,n$, taking values in $\{0,1\}$. The proportion of ones can be estimated as $\hat{p}_1=\frac{1}{n}\sum_{i=1}^n x_{bi}$, with variance $\hat{p}_1(1-\hat{p}_1)/n$. Since we are interested in the difference between groups, we assume that the outcome is a continuous variable. However, extensions to categorical outcomes can be achieved by incorporating methods such as logistic regression.

Let $(\X_{T,i}, Y_{T,i}, Z_{T,i})$ and $(\X_{S_k,i}, Y_{S_k,i}, Z_{S_k,i})$ follow the distributions $P_{T}(\X, Y, Z)$ and $P_{\cM}(\X, Y, Z)$, respectively. To utilize information from $\cM$, we assume a covariate shift between the target set $T$ and the meta-analysis set $\cM$ as follows \citep{shimodaira2000improving}:
\begin{enumerate}
    \item[(A2)] \textbf{Covariate shift}: Suppose that $(\X_{T,i}, Y_{T,i}, Z_{T,i}) \sim P_{T}(\X, Y, Z)$ and $(\X_{\cM,i}, Y_{\cM,i}, Z_{\cM,i}) \sim P_{\cM}(\X, Y, Z)$. Then,
    \begin{align*}
        P_{T}(Y \mid \X, Z) &= P_{\cM}(Y \mid \X, Z).
    \end{align*}
\end{enumerate}
This assumption implies that the outcomes for each group in $\cM$ and $T$ follow the same conditional distribution given the covariates, while allowing the marginal distributions of the covariates themselves to differ (i.e., $P_T(\X, Z) \neq P_{\cM}(\X, Z)$ is permitted). Assumption (A2) is commonly used in transfer learning \citep{gretton2006kernel, huang2006correcting, wang2014active}. In this study, assumption (A2) enables us to transfer the outcome information from $\cM$ to the target population $T$ by using the ratio of the joint density of $(\X,Z)$ between $T$ and $\cM$. The following two examples are typical cases where we apply the InMASS method:
\begin{example}[Integrate to target trial for historical control]  
\label{ex-hist}  
    We aim to estimate the CATE $\delta_T$ of an ongoing trial $T$ using the results of a meta-analysis. The available data consist of the IPD $\{y_{Ti}, \x_{Ti}\}_{i=1,\dots,n_T}$ for $T$ and the AD $\{n_{kj}, \bar{y}_{kj}, \bar{\x}_{kj}, \sigma_{\bar{y}_{kj}}^2, \bsi_{\bar{\x}_{kj}}^2\}_{j=0,1,\; k=1,\dots,K}$ for the trials in the meta-analysis $S_1, \dots, S_K$. To satisfy the planned statistical power and significance level, we borrow $S_1, \dots, S_K$ as historical controls for $T$.
\end{example}  

\begin{example}[Integrate to target trial for both treatment and control]  
\label{ex-tc}  
    We aim to estimate the CATE $\delta_T$ for a trial $T$ of interest, such as a clinical or observational study. The available IPD and AD are the same as in Example \ref{ex-hist}; specifically, the IPD for $T$ is available, while only AD is available for $S_1, \dots, S_K$. In this scenario, information from $S_1, \dots, S_K$ is used for both the treatment and control groups of the target trial $T$. 
\end{example}

\section{Methods}\label{sec-method}

\subsection{Overview of the proposed method}
InMASS consists of four steps: (1) meta-analysis, (2) IPD reconstruction, (3) density ratio estimation, and (4) weighted regression. 
The flowchart for InMASS is displayed in Figure \ref{fig-flowchart}. The method enables the estimation of $\delta_T$ by reconstructing IPD outcomes and covariates from the AD of the meta-analysis.
\begin{enumerate}
    \item[(1)] \textbf{Meta-analysis}: In this step, we estimate the regression parameters required for IPD reconstruction. Since this study focuses on estimating $\delta_T$, methods applicable to continuous variables—such as restricted maximum likelihood, the method of moments, or other conservative approaches—can be used \citep{dersimonian1986meta, hartung1999alternative, raudenbush2009analyzing}.
    \item[(2)] \textbf{IPD Reconstruction}: First, covariates are reconstructed using the reported means and variances for each trial, followed by outcome reconstruction using the meta-analytic regression parameters. As shown in Section \ref{sec-ripd}, the reconstructed covariates and outcomes are guaranteed to match the reported means and second moments.
    \item[(3)] \textbf{Density Ratio Estimation}: The distribution of the reconstructed IPD from $\cM$ differs from that of the target set $T$. The density ratio between $T$ and $\cS = (\cM, T)$ is used to adjust for this difference in order to estimate $\delta_T$. In this step, logistic regression is applied to estimate the density ratio under the covariate shift assumption. As demonstrated in Section \ref{sec-iwlm}, if valid density ratios are obtained, matching the second moments is sufficient for the reconstruction of IPD from $\cM$.
    \item[(4)] \textbf{Weighted Regression}: Finally, the reconstructed IPD from the meta-analysis and the IPD from the target trial are integrated, and a weighted regression analysis is performed to estimate $\delta_T$.
\end{enumerate}
Through these four steps, InMASS provides a framework for integrating information from both existing meta-analyses and the target trial to efficiently estimate the CATE for the target trial $T$.

\begin{figure}[H]
    \centering
    \includegraphics[width=0.9\linewidth]{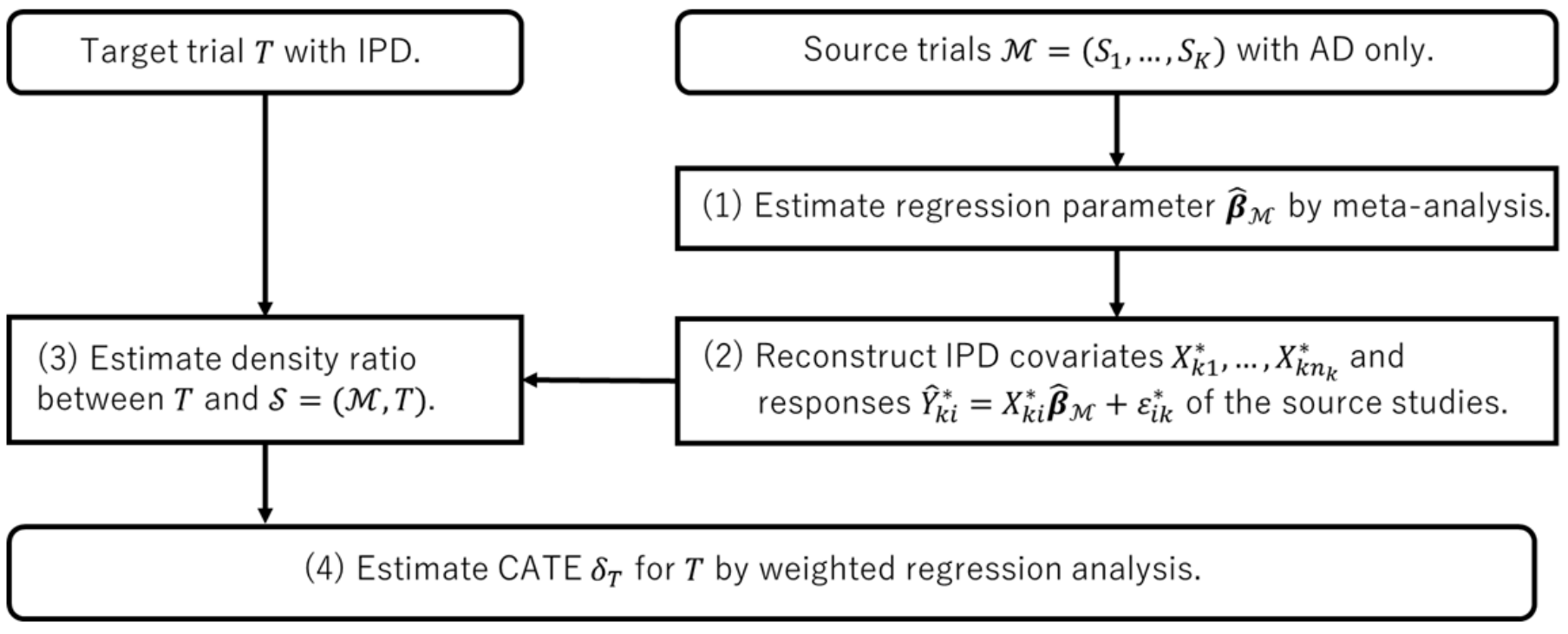}
    \caption{Flowchart of InMASS. AD: aggregate data, IPD: individual participant data, CATE: conditional average treatment effect.}
    \label{fig-flowchart}
\end{figure}

\subsection{Reconstructing individual participant data from meta-analysis}\label{sec-ripd}
We perform a random-effects meta-analysis using $K$ trials \citep{dersimonian1986meta} for reconstructing IPD. The estimated treatment effect of the $k$-th trial, defined as $\hat{\delta}_{S_k} = \bar{y}_{k1} - \bar{y}_{k0}$, is assumed to be normally distributed with mean $\delta_{S_k}$ and within-trial variance $\sigma_{S_k}^2$, i.e., $\hat{\delta}_{S_k} \sim N(\delta_{S_k}, \sigma_{S_k}^2)$.
In random-effects meta-analysis, the true treatment effect $\delta_{S_k}$ for the $k$-th trial is assumed to follow $\delta_{S_k} \sim N(\delta_{\cM}, \tau^2)$, where $\delta_{\cM}$ and $\tau^2$ represent the true overall treatment effect of $\cM$ and the between-trial variance, respectively. Marginalizing over $\delta_{S_k}$ yields that the $k$-th trial treatment estimator $\hat{\delta}_{S_k}$ is distributed as $\hat{\delta}_{S_k} \sim N(\delta_{\cM}, \sigma_{S_k}^2 + \tau^2)$.
A simple extension of the random-effects model to meta-regression replaces the overall treatment effect $\delta_{\cM}$ with a linear predictor $X_k \bbe$, i.e.,
\begin{align}
    \label{eq-rm-reg}
    \hat{\delta}_{S_k} \sim N(\bar{\x}_k \bbe_{\cM}, \sigma_{S_k}^2 + \tau^2),
\end{align}
where $\bar{\x}_k = (\bar{\x}_{k1}, \bar{\x}_{k0})^T$ is a $2 \times p$ matrix of covariate means and $\bbe_{\cM} = (\beta_{\cM 1}, \dots, \beta_{\cM p})^T$ is a parameter vector for the meta-analysis population $\cM$.

Under the meta-regression model \eqref{eq-rm-reg}, various methods have been proposed for estimating the unknown parameters $(\bbe_{\cM}, \tau^2)$ \citep{dersimonian1986meta, veroniki2016methods}. While these methods tend to underestimate the between-trial variance $\tau^2$ when the number of trials is small \citep{li1994bias, viechtbauer2005bias}, the regression parameter $\bbe_{\cM}$ remains both unbiased and consistent in estimation.
Therefore, we assume that the estimator of the regression parameter, $\hat{\bbe}_{\cM}$, is unbiased and consistent, i.e., 
\[
{\rm E}[\hat{\bbe}_{\cM}] = \bbe_{\cM} \quad \text{and} \quad {\rm V}[\hat{\bbe}_{\cM}] = \Si_{\hat{\bbe}_{\cM}} \to \mathbf{0} \quad \text{as} \quad K \to \infty,
\]
where $\Si_{\hat{\bbe}_{\cM}}$ denotes the variance of $\hat{\bbe}_{\cM}$.

We reconstruct the individual covariate $\X_{ki}^*$, for $i=1,\dots,n_k$, from either the true distribution $F_{\X}(\cdot; \bmu_k, \bsi_{\bmu_k}^2)$—which has mean $\bmu_k$ and variance $\bsi_{\bmu_k}^2 \I_p$—or from the estimated distribution $\hat{F}_{\X}(\cdot; \bar{\x}_k, \hat{\bsi}_{\bar{\x}_k}^2)$. The observed outcome for subject $i$ is reconstructed according to the model
\begin{align}
\label{eq-ripd}
    \hat{Y}_{ki}^* &= \X_{ki}^* \hat{\bbe}_{\cM} + \varepsilon_{ki}^*,
\end{align}
where the error term $\varepsilon_{ki}^*$ is assumed to be normally distributed with mean $0$ and variance $\hat{\sigma}_k^2 - (\hat{\bbe}_{\cM}^2)^T \hat{\bsi}_{\bar{\x}_k}^2$.
Since the true distribution $F_{\X}(\cdot; \bmu_k, \bsi_{\bmu_k}^2)$ is generally not available in practice, we use the estimated distribution $\hat{F}_{\X}(\cdot; \bar{\x}_k, \hat{\bsi}_{\bar{\x}_k}^2)$ to reconstruct covariates. While the choice of $\hat{F}_{\X}(\cdot)$ depends on whether $\X$ is continuous or discrete, any distribution capable of reproducing the mean and variance may be employed. For instance, a normal distribution $N(\bar{\x}_k, \hat{\bsi}_{\bar{\x}_k}^2)$ is appropriate for continuous variables, whereas a Bernoulli distribution, $\text{Bernoulli}(\bar{\x}_k)$, can be used for binary variables. This assumption implies that the sample means and variances for each trial converge to their true values:
\begin{enumerate}
    \item[(C1)] For any trial $S_k \in \cM$, $\bar{\x}_k \to \bmu_k$ and $\hat{\bsi}_{\bar{\x}_k}^2 \to \bsi_{\bmu_k}^2$ in probability as $n_k \to \infty$.
\end{enumerate}
Then, we can reconstruct the outcome in trial $S_k$ consistently up to the second moment.

\begin{prop}[Restorability of individual participant data]\label{prop1}
    Suppose that the true covariate distribution $F_{\X}$ is known. Then, the resampled individual treatment effect $\hat{Y}_{ki}^*$ defined in \eqref{eq-ripd} has mean $\bmu_k \hat{\bbe}_{\cM}$ and variance $\sigma_k^2$. Furthermore, if $F_{\X}$ is unknown and assumption (C1) is satisfied, then the resampled individual treatment effect $\hat{Y}_{ki}^*$ obtained using the estimated distribution $\hat{F}_{\X}(\cdot; \bar{\x}_k, \hat{\bsi}_{\bar{\x}_k}^2)$ converges in probability, as $n_k\to\infty$, to a distribution with mean $\bmu_k \hat{\bbe}_{\cM}$ and variance $\sigma_k^2$.
\end{prop}
\begin{proof}
    The proof is provided in Appendix \ref{app}.
\end{proof}
Although the consistency of higher moments (third or higher) is unclear, Proposition \ref{prop1} is sufficient for estimating $\delta_T$ and its variance, as demonstrated in Section \ref{sec-iwlm}. When estimating $\bbe_{\cM}$ using standard meta-analysis methods, the expectation and variance of $\hat{\bbe}_{\cM}$ satisfy 
${\rm E}[\hat{\bbe}_{\cM}] = \bbe_{\cM}$ and ${\rm V}[\hat{\bbe}_{\cM}] \to \mathbf{0}$ as $K \to \infty$.
Therefore, ${\rm E}[\hat{Y}_{ki}^*] = \bmu_k \bbe_{\cM}$ and ${\rm V}[\hat{Y}_{ki}^*] \to \sigma_k^2$ as $K \to \infty$.
In other words, the expectation obtained from the reconstructed IPD is unbiased, and its variance converges to the true variance of each trial when the number of studies is sufficiently large.

\subsection{Importance weighting for estimating CATE}\label{sec-iwlm}
The importance weighting is useful for transfer learning \citep{byrd2019effect}, as it allows any loss function defined on the source data to be transferred to the target domain. Let $L(\x,y,z)$ be any loss function. Under the covariate shift assumption (A2), we have
\begin{align*}
    {\rm E}_{T}[L(\X,Y,Z)] 
    &= \int_{T} L(\x,y,z) \, dP_T(\x,y,z) \\
    &= \int_{T} L(\x,y,z) \, dP_{T}(y \mid \x,z) \, dP_{T}(\x,z) \, \frac{dP_{\cS}(\x,z)}{dP_{\cS}(\x,z)} \\
    &\quad \text{(since } P_{T}(y \mid \x,z)=P_{\cS}(y \mid \x,z) \text{ from (A2))} \\
    &= \int_{T} \frac{dP_{T}(\x,z)}{dP_{\cS}(\x,z)} \, L(\x,y,z) \, dP_{\cS}(\x,y,z) \\
    &= {\rm E}_{\cS}\left[\frac{dP_{T}(\X,Z)}{dP_{\cS}(\X,Z)} \, L(\X,Y,Z)\right],
\end{align*}
where $dP_{D}(\cdot)$ denotes the probability density function of $(\X,Z)$ for domain $D$, and $dP_{D}(\cdot \mid \x,z)$ denotes the probability density function of $Y$ given $(\x,z)$ for domain $D$.
Suppose that the importance weight is defined by $w_{Si} = \frac{dP_{T}(\x_{Si}, z_{Si})}{dP_{\cS}(\x_{Si}, z_{Si})}$, and consider the loss function $L(\x,y,z) = \left(y-\delta_T z - \x^T \bbe \right)^2$.
Then, the least squares estimator of $\delta_T$ can be obtained.
We further define a boundedness condition for the density ratio to apply importance weighting in estimating $\delta_T$:
\begin{enumerate}
    \item[(C2)] \textbf{Boundedness of Density Ratio:} For any subject set $S \in \cS$, 
    \[
    \sup_{i \in S} \left| \frac{dP_T(\x_{Si}, z_{Si})}{dP_{\cS}(\x_{Si}, z_{Si})} \right| < \infty.
    \]
\end{enumerate}
Under condition (C2), the group difference $\delta_T$ can be estimated as follows:
\begin{align}
\label{deltaR}
    \hat{\delta}_{R} &= \bar{Y}_1 - \bar{Y}_0,
\end{align}
where the weighted mean of each group $j \, (=0,1)$ is defined by
\begin{align*}
    \bar{Y}_j &= \frac{\sum_{S \in \cS} \sum_{i=1}^{n_S} w_{Si}\,\mathbbm{1}(z_{Si}=j)\,\hat{Y}_{Si}}{\tilde{N}_j}.
\end{align*}
Here, $\hat{Y}_{Si}$ denotes the outcome for subject $i$ in set $S$, where 
\[
\hat{Y}_{Si} = 
\begin{cases}
    y_{Ti}, & \text{if } S = T, \\
    \hat{Y}_{ki}^*, & \text{if } S = S_k,
\end{cases}
\]
and the weighted sample size for group $j$ is defined as 
$\tilde{N}_j = \sum_{S \in \cS} \tilde{n}_{Sj}$, with $\tilde{n}_{Sj} = \sum_{i=1}^{n_S} w_{Si}\,\mathbbm{1}(z_{Si}=j)$.
In this context, $z_{Si}$ is an indicator variable such that $z_{Si}=1$ if subject $i$ belongs to the treatment group and $z_{Si}=0$ if subject $i$ belongs to the control group.
The variance of $\hat{\delta}_R$ is given by $\sigma_R^2 = \tilde{N}_1^{-1}\,\sigma_{R1}^2 + \tilde{N}_0^{-1}\,\sigma_{R0}^2$, where the variance for each group $j$ is defined as
\begin{align}
\label{sigma2R}
    \sigma_{Rj}^2 &= \frac{\sum_{S \in \cS} \sum_{i=1}^{n_S} w_{Si}^2\,\mathbbm{1}(z_{Si}=j)\,\left(\hat{Y}_{Si} - {\rm E}_T[Y\mid Z=j]\right)^2}{\tilde{N}_j}.
\end{align}
Since ${\rm E}_T[Y\mid Z=j]$ is generally unknown, it is replaced by its estimator $\bar{Y}_j$, and the resulting variance estimator is denoted by $\hat{\sigma}_{Rj}^2$. Consequently, the estimators $\hat{\delta}_R$ and $\hat{\sigma}_R^2$ are consistent.

\begin{theorem}[Consistency of $\hat{\delta}_{R}$ and $\hat{\sigma}_{R}^2$]\label{th1}
    Under the assumptions (A1) and (A2) and the conditions (C1) and (C2), the estimator $\hat{\delta}_R$ defined in \eqref{deltaR} is a consistent estimator of $\delta_T$, i.e.,
    \[
    \hat{\delta}_R \to \delta_T \quad \text{in probability as} \quad N_j = \sum_{S \in \cS} n_{Sj} \to \infty, \quad j=0,1.
    \]
    Moreover, the variance estimator $\hat{\sigma}_{R}^2$ converges to $\sigma_R^2$. 

    Furthermore, if $\delta_T < \infty$ and $\sigma_T^2 < \infty$, then 
    \begin{align}
        \frac{\hat{\delta}_R - \delta_T}{\sqrt{\tilde{N}_1^{-1}\hat{\sigma}_{R1}^2 + \tilde{N}_0^{-1}\hat{\sigma}_{R0}^2}}
    \end{align}
    converges in distribution to the standard normal distribution as $N_j \to \infty$ for $j=0,1$.
\end{theorem}

The estimator $\hat{\delta}_R$ is consistent for $\delta_T$ under appropriate conditions and converges to a normal distribution when the numbers of participants in both the treatment and control groups are sufficiently large. 
Notably, it is not necessary for the number of participants in each trial $S \in \cS$ to increase indefinitely (i.e., $n_{Sj} \to \infty$ for all $S$). Rather, it suffices for the total number of participants across $\cS$ to grow without bound (i.e., $N_j \to \infty$).
Thus, even if the number of participants in the target trial $T$ is insufficient, consistency and asymptotic normality still hold as long as the number of participants in the other trials is sufficiently large.

The following theorem shows that the estimator $\hat{\delta}_R$ based on the source population $\cS$ achieves a smaller variance than the estimator $\hat{\delta}_T$ based solely on the target trial $T$.

\begin{theorem}\label{th2}
    Under assumptions (A1) and (A2) and conditions (C1) and (C2),
    \begin{align}\label{mse-ratio}
        \frac{{\rm V}[\hat{\delta}_{R}]}{{\rm V}[\hat{\delta}_{T}]} = O\left( \frac{n_T}{\tilde{N}} \right) \quad \text{and} \quad \frac{{\rm V}[\hat{\delta}_{R}]}{{\rm V}[\hat{\delta}_{T}]} = O\left(K^{-1}\right).
    \end{align}
    Moreover, if $w_{S_k i} = 0$ for all $k=1,\dots,K$ and $w_{Ti}=1$, then ${\rm V}[\hat{\delta}_{R}] / {\rm V}[\hat{\delta}_{T}] = 1$.
\end{theorem}

\begin{proof}[Proof of Theorems \ref{th1} and \ref{th2}]
    The proofs are provided in Appendix \ref{app}.
\end{proof}

From Theorem \ref{th2}, if a sufficiently large number of subjects in population $\cS$ is available, the variance of the estimator obtained using InMASS will be smaller than that of the estimator based solely on the target trial $T$. This indicates that even when it is difficult to increase the number of participants in $T$, InMASS can estimate $\delta_T$ more efficiently than $\hat{\delta}_T$ by increasing either the number of meta-analyses $K$ or the number of participants within $\cM$. It is noted that the variance comparison in Theorem \ref{th2} depends on the ratio of $\tilde{N}$ to $n_T$, rather than on the total number of participants $N$ in $\cS$. This excludes special cases where only participants with density ratios $w_{Si}=0$ are increased. Additionally, from the consistency result in Theorem \ref{th1}, we have 
\[
\frac{{\rm MSE}(\hat{\delta}_{R})}{{\rm MSE}(\hat{\delta}_{T})} \to \frac{{\rm V}[\hat{\delta}_{R}]}{{\rm V}[\hat{\delta}_{T}]}
\]
in probability as $N_j \to \infty$ for $j=0,1$.

\subsection{Estimating the weights of density ratio}
The density ratio $w_{Si}$ is unknown in practice and is estimated by various methods. Common approaches include logistic regression, minimization of the Kullback-Leibler divergence, and least-squares fitting \citep{qin1998inferences, cheng2004semiparametric, kanamori2009least}. In this study, we apply logistic regression, a widely used approach in causal inference for estimating propensity scores. We assume the following logistic regression model:
\begin{align}
\label{eq-logit}
    \log \pi(\X) &= \log \frac{P(s \in T \mid \X)}{P(s \in \cS \mid \X)} = \X \bal.
\end{align}
Then, the density ratio is estimated as 
\begin{align*}
    \hat{w}(\X_{Si}^*) &= \frac{N}{n_T} \hat{\pi}(\X_{Si}^*),
\end{align*}
where $\hat{\pi}(\cdot)$ is the estimator obtained from the logistic model \eqref{eq-logit}. While this type of density ratio estimation is commonly applied in fields such as domain adaptation and sample selection bias correction \citep{daume2006domain, cuddeback2004detecting}, it is noted that the estimate is affected by the ratio of the number of subjects between the target trial and the meta-analysis trials.

\subsection{Extend to weighted regression for CATE}
$\hat{\delta}_R$ and $\hat{\sigma}_R^2$ are formally univariate versions of weighted regression analysis and can be extended to the following weighted linear regression model:
\begin{align}
\label{eq-reg-model}
    \hat{\Y}_{\cS} = \X_{\cS} \bbe_{T} + \bep_{\cS},
\end{align}
where $\hat{\Y}_{\cS} = (y_{T1}, \dots, y_{Tn_T}, \hat{Y}_{11}^*, \dots, \hat{Y}_{Kn_K}^*)^T$ is the $N \times 1$ outcome vector with $N = n_T + \sum_{k=1}^K n_{S_k}$. 
The design matrix $\X_{\cS}$ is an $N \times (p+2)$ matrix given by $\X_{\cS} = (\mathbf{1}, \z, \x)$, where $\mathbf{1} = (1, \dots, 1)^T$ is an $N \times 1$ vector, $\z = (z_{T1}, \dots, z_{Tn_T}, z_{S_1 1}, \dots, z_{S_K n_{S_K}})^T$ is an $N \times 1$ indicator vector (with $z_{Si} = 1$ for treatment and $z_{Si} = 0$ for control), $\x = (\x_{T1}, \dots, \x_{Tn_T}, \x_{S_1 1}, \dots, \x_{S_K n_{S_K}})^T$ is an $N \times p$ covariate matrix.
The error vector $\bep_{\cS} = (\ep_{T1}, \dots, \ep_{Tn_T}, \ep_{S_11}, \dots, \ep_{S_K n_{S_K}})^T$ is assumed to be independent of $\X_{\cS}$ and its variance-covariance matrix is denoted by $\bSi_{\ep}$.
The parameter vector is $\bbe_{T} = (\beta_0, \delta_T, \beta_1, \dots, \beta_p)^T$, where $\delta_T$ denotes the group difference of interest. 
For each subject $i$ in a trial $S \in \cS$, $\x_{Si}$ is the $p \times 1$ covariate vector, and the error term $\ep_{Si}$ is independently distributed with mean $0$ and variance $\sigma_S^2$. Note that while the error terms are independent across subjects, they are not identically distributed due to the heterogeneity in variances across trials.

In this setup, $\bbe_{T}$ is estimated by weighted least squares as
\[
\hat{\bbe}_T = (\X_{\cS}^T \W \X_{\cS})^{-1} \X_{\cS}^T \W \hat{\Y}_{\cS},
\]
where $\W = \diag(w_{T1}, \dots, w_{Tn_T}, w_{S_1 1}, \dots, w_{S_K n_{S_K}})$ is the $N \times N$ weight matrix with density ratios on the diagonal. The variance-covariance matrix of $\hat{\bbe}_T$ is computed as
\[
\bSi_{\hat{\bbe}_T} = (\X_{\cS}^T \W \X_{\cS})^{-1} \Bigl(\X_{\cS}^T \W \bSi_{\ep} \W \X_{\cS}\Bigr) (\X_{\cS}^T \W \X_{\cS})^{-1}.
\]
Since $\bSi_{\ep}$ is unknown in practice, we apply the heteroskedasticity-consistent covariance matrix estimator proposed by \cite{white1980heteroskedasticity} under the covariate shift assumption, namely,
\[
\hat{\bSi}_{\ep} = \diag\Bigl(\W^2 (\hat{\Y}_{\cS} - \X_{\cS} \hat{\bbe}_T)^2\Bigr).
\]
For the consistency of $\hat{\bbe}_T$, we assume the following regularity condition:
\begin{enumerate}
    \item[(C3)] $\frac{1}{N} (\X_{\cS}^T \W \X_{\cS})^{-1} \to \Q \quad \text{in probability as } N \to \infty$, where $\Q$ is a $(p+2) \times (p+2)$ positive-definite matrix.
\end{enumerate}

\begin{theorem}
\label{th-central-limit-theorem}
    Under assumptions (A1) and (A2) and conditions (C1) and (C2), if the model \eqref{eq-reg-model} is correctly specified, then
    \[
    {\rm E}_{\cS}[\hat{\bbe}_T] = \bbe_T.
    \]
    Furthermore, if condition (C3) holds, then the variance estimator satisfies $\hat{\bSi}_{\hat{\bbe}_T} \to \bSi_{\hat{\bbe}_T}$ in probability, and $\sqrt{N} (\hat{\bbe}_T - \bbe_T)$ converges in distribution to $N\bigl(\mathbf{0}, \bSi_{\hat{\bbe}_T}\bigr)$ as $N \to \infty$.
\end{theorem}
\begin{proof}[Proof of Theorem \ref{th-central-limit-theorem}]
    The proofs of these results are provided in Appendix \ref{app}.
\end{proof}

This extension allows for considering relationships between variables, similar to standard linear regression, compared to the univariate case. It enables adjustment for known confounding factors and elucidates which variables have a greater impact on the outcome based on their estimated coefficients. However, increasing the number of variables increases the complexity of the model, which may result in overfitting or multicollinearity issues.
Although this is a general problem in regression analysis, Theorem \ref{th-central-limit-theorem} provides the asymptotic distribution, thereby enabling model selection using criteria such as the Akaike Information Criterion \citep{akaike1973information}.

\section{Numerical studies}\label{sec-num-result}
In Section \ref{sec-simulation}, we evaluate the performance of InMASS through simulations. In Section \ref{sec-case-study}, we apply InMASS to a trial replicated from an existing meta-analysis.

\subsection{Simulation}\label{sec-simulation}
The simulation compares the estimation of $\delta_T$ using only the target trial $T$ with the estimation using InMASS, which is based on the IPD of $T$ and the AD of $\cM$. We assume the following individual-level model for participants $i$ in trials $S_k \in \cM$, for $k = 1, \dots, K$, used in the meta-analysis:
\begin{align}\label{eq-sim-model}
    Y_{ki} &= \beta_0 + \delta_T z_{ki} + \beta_1 x_{1ki} + \beta_2 z_{ki} x_{1ki} + \epsilon_{ki}, \quad k=1,\dots, K,\quad i=1, \dots, n_k,
\end{align}
where $z_{ki}$ is the group indicator (with $z_{ki}=1$ for subjects in the treatment group and $z_{ki}=0$ for those in the control group), $x_{1ki}$ is an observable covariate that is normally distributed with mean $\mu_{1k} = 4(k-1)/(K-1) - 1$ and variance $1$, the error term $\epsilon_{ki}$ for subject $i$ in trial $S_k \in \cS$ follows standard normal distribution, and the parameter values are $(\delta_T, \beta_0, \beta_1, \beta_2) = (2, 1, -1, 0.5)$. The model \eqref{eq-sim-model} represents an ideal case in which the covariates are normally distributed.
In addition to this setting, we consider scenarios to examine whether the claim of Theorem \ref{th1}—that CATE estimation remains valid as long as the second moments are reconstructed, even when the observed covariates deviate from normality—holds numerically. Specifically, we examine two scenarios: one in which the observed covariate is normally distributed and one in which it is skewed. The two scenarios are as follows:
\begin{enumerate}
    \item[(1)] Model \eqref{eq-sim-model} with $x_{1ki} \sim N(\mu_{1k}, 1)$.
    \item[(2)] Model \eqref{eq-sim-model} with $x_{1ki} \sim \chi_2^2 / 2 + \mu_{1k} - 1$.
\end{enumerate}
Here, $\chi_r^2$ denotes a chi-squared random variable with $r$ degrees of freedom.

The number of studies $K$ and the number of participants $n_{S_k}$ are chosen to reflect realistic conditions. Specifically, the number of studies ranges from relatively small to moderately large ($K=5, 10, 30$). The number of participants in each meta-analyzed trial $S \in \cM$ is set to be the integer part of a value sampled from a uniform distribution $U(n, 4n)$ with $n=20, 40, 100$, and the participants are evenly split between the treatment and control groups. The individual participant model for the target trial $T$ is assumed to be essentially the same as model \eqref{eq-sim-model}, except that $\mu_{1k} = 0$. This setting ensures that the covariate shift assumption is satisfied.
The target trial considers three allocation scenarios: 1:1 allocation, 3:1 allocation, and a single-arm treatment group. In these scenarios, the numbers of participants in the treatment and control groups are set as $(n_1, n_0) = (n/2, n/2), (3n/4, n/4), (n, 0)$, respectively.

Since the primary interest is the group difference, we focus on the estimators $\hat{\delta}_T$ and $\hat{\delta}_R$ for $\delta_T$. We selected two models for the regression analysis of the target trial and the weighted regression in InMASS. The first model specifies the true model given in \eqref{eq-sim-model}, while the second model is misspecified by including only the intercept and treatment arm as explanatory variables. 
Under the assumptions of model \eqref{eq-sim-model}, the information available for meta-analysis from each group in each trial includes the sample size, the mean and variance of the outcome, and the mean and variance of the covariate $x_{1ki}$, namely, $\{n_{kj}, \bar{y}_{kj}, \bar{x}_{1kj}, \sigma_{\bar{Y}_{kj}}^2, \sigma_{\bar{x}_{1kj}}^2\}_{j=0,1}$.

We evaluate the performance of InMASS by comparing its estimation of group differences to the estimation based solely on information from the target trial under both 1:1 and 3:1 allocation scenarios. InMASS employs two strategies: (1) for the 1:1 allocation, it borrows as much information as possible from $\cM$, and (2) for the 3:1 allocation and a single-arm treatment group, it borrows only the information of the control group from $\cM$. To avoid computational estimation failures due to Fisher information approaching zero, InMASS employs the DerSimonian-Laird method \citep{dersimonian1986meta} instead of the restricted maximum likelihood method for meta-analysis. We estimated density ratios using the logistic regression model \eqref{eq-logit}, including the intercept, covariates, and their squared terms as explanatory variables. Data generation and analysis in the simulations were performed 10,000 times. We compare the mean squared error (MSE) and the power of each method. The power is computed as the probability that the 95\% confidence interval for $\delta_T$ does not contain the null value $\delta_0$, where $\delta_0$ is varied from $\delta_T$ to $\delta_T + 2$ in increments of 0.1. When $\delta_0 = \delta_T$, this probability represents the type I error rate; for $\delta_0 > \delta_T$, a higher probability indicates greater statistical power.

The MSE for $K=10$ is shown in Figure \ref{fig-mse}, and the results for other scenarios are provided in Appendix Figures \ref{fig-1to1-mse}--\ref{fig-controlAD-mse}. Compared to using only the target trial with 1:1 allocation, InMASS significantly improved the MSE, which supports Theorem \ref{th2}. This result was consistent regardless of model specification and covariate normality, with a particularly large improvement in MSE when the sample size was small.
In the case of 3:1 allocation and a treatment single group, InMASS is applied to borrow subjects from the control group. The MSE for InMASS applied to 3:1 allocation and a single-arm treatment group is comparable to that obtained from a 1:1 allocation trial. When the model is misspecified and the number of subjects is small, InMASS yields a larger MSE than a 3:1 allocation trial alone. However, as the number of subjects increases, the MSE of InMASS approaches that of the 1:1 allocation trial. By allocating $n$ participants with a higher proportion to the treatment group, the overall amount of available information increases, even if subjects can only be borrowed for the control group. As a result, InMASS improves the MSE and is robust to the distribution of covariates.

\begin{figure}[H]
    \centering
    \includegraphics[width=0.9\linewidth]{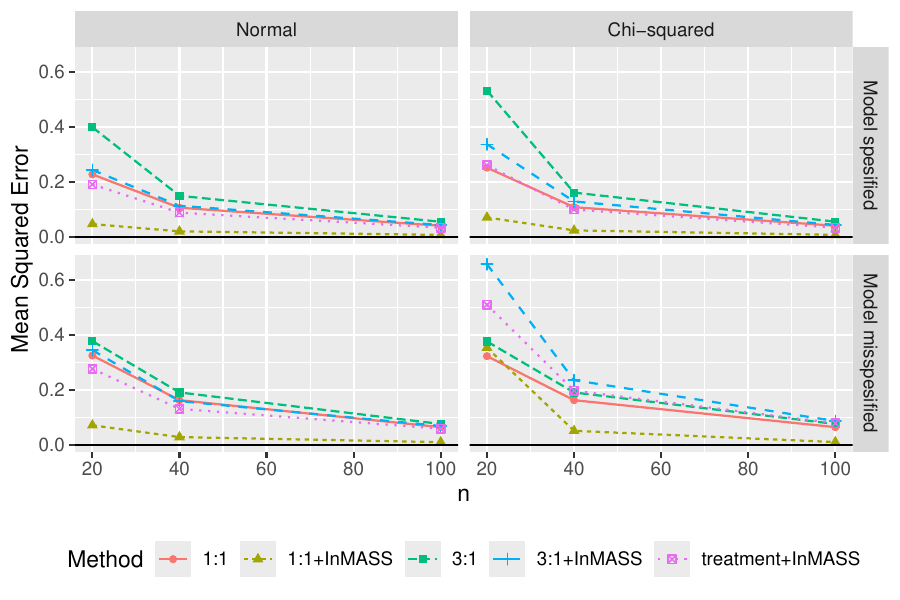}
    \caption{Mean squared error of CATE for the target population with $K=10$.}
    \label{fig-mse}
\end{figure}

Figure \ref{fig-power1} shows the power plot for $\delta_T$ with $K=10$, and the results for the other scenarios are provided in Appendix Figures \ref{fig-1to1-power1}--\ref{fig-controlAD-mid-power2}. As with MSE, InMASS significantly improved power compared to using only the target trial with 1:1 allocation, regardless of the model specification and the covariate distribution. In the case of 3:1 allocation and a single-arm treatment group, InMASS achieved similar power to that of the 1:1 allocation trial, despite borrowing information only from the control group. 
When the model is misspecified, applying InMASS to the control group resulted in achieving a power of 0.8 or 0.9 more quickly than not using InMASS. However, because the InMASS estimator of $\delta_T$ was biased, the 3:1 allocation without InMASS showed higher power than the 1:1 allocation for values of power up to approximately 0.5.
On the other hand, due to the increased variance of the weighted regression under model misspecification, the type I error did not inflate and remained controlled below 0.05.

In summary, InMASS clearly improves both MSE and power by utilizing the AD from meta-analysis. In cases such as 3:1 allocation or single-arm trials—where only control group information is borrowed—increased overall information leads to improved MSE and power, while the type I error rate remains controlled even when the model is misspecified and the covariate follows a non-normal distribution.

\begin{figure}[H]
    \centering
    \includegraphics[width=0.8\linewidth]{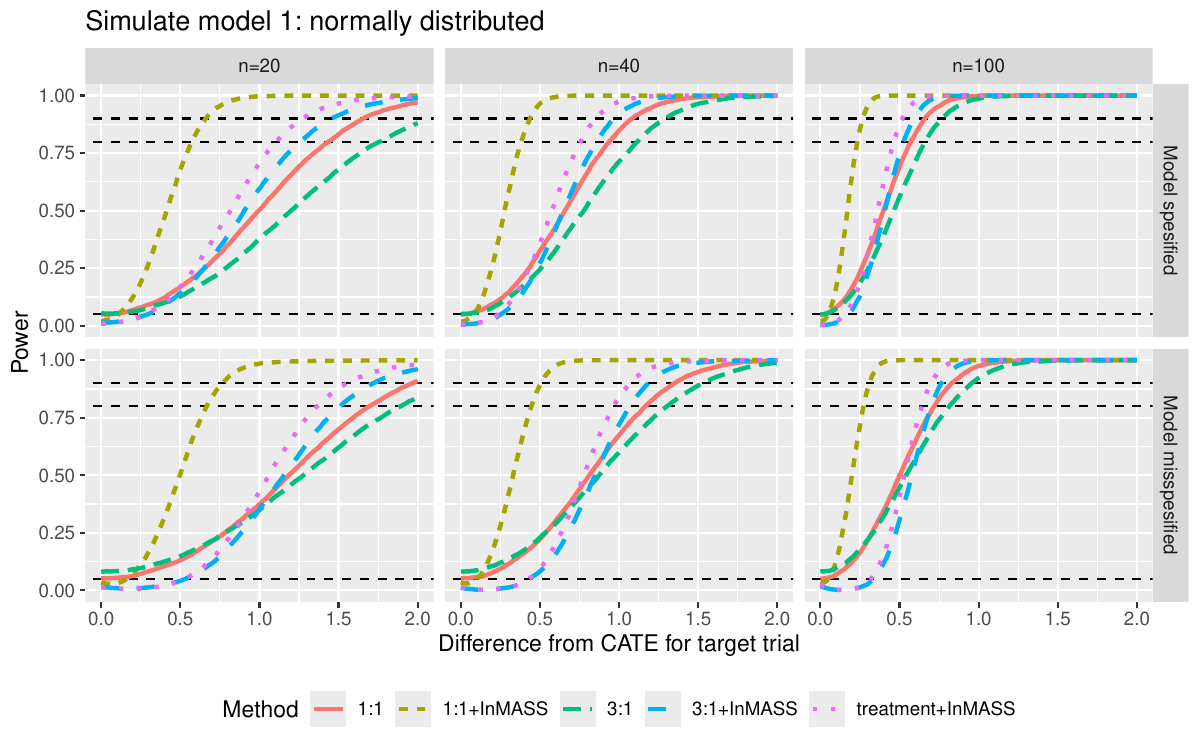}
    \includegraphics[width=0.8\linewidth]{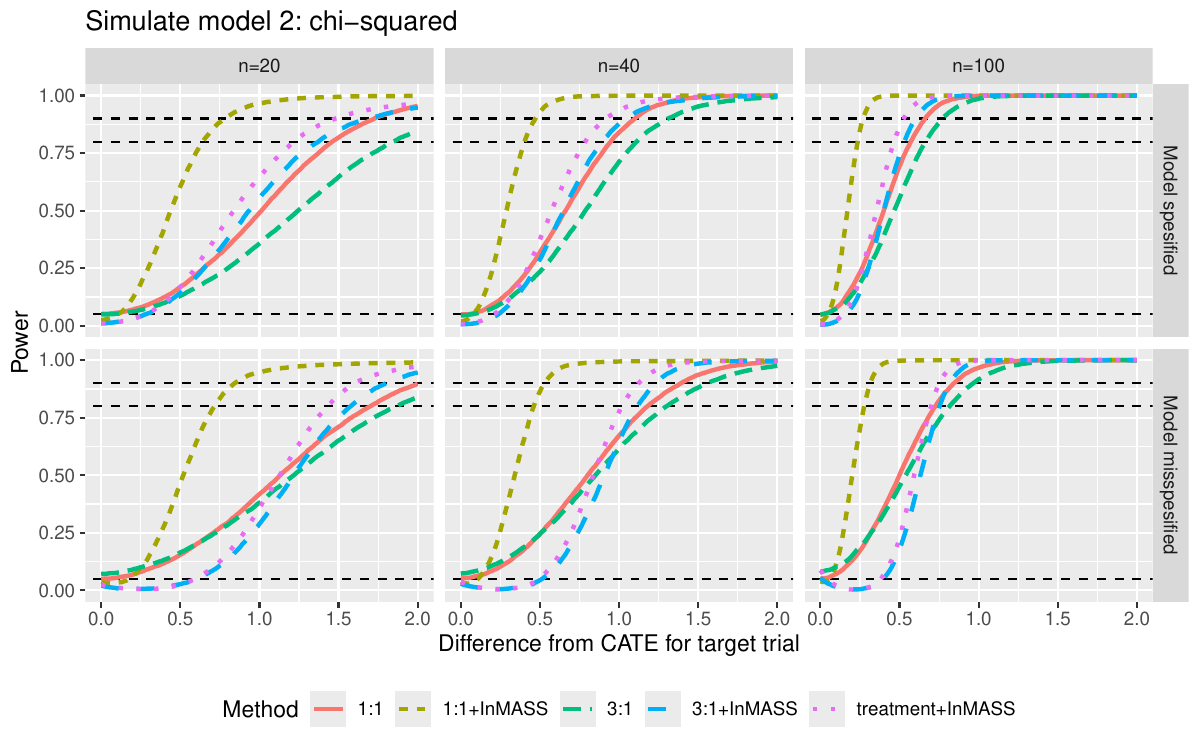}
    \caption{Power of each method for simulation model \eqref{eq-sim-model} with $K=10$.}
    \label{fig-power1}
\end{figure}

\subsection{Case study}\label{sec-case-study}
As an application of InMASS, we present an example in which meta-analysis is used to supplement the sample size in a small-sample trial. Here, we consider the meta-analysis by \cite{sanguankeo2015effects}, which conducted a meta-analysis of five trials and reported the mean difference in annual estimated glomerular filtration rate (eGFR) change between the Statin and Control groups for a particular outcome. The meta-analysis also reported baseline eGFR and the number of participants for each trial. The outcome and covariate data used in this study are summarized in Table \ref{tab-rda-meta}. In this study, we selected Sawara 2008 as the target trial. We simulated a hypothetical clinical trial based on the settings of the target trial and considered borrowing participants from the meta-analysis of the remaining four trials for analysis. The target trial was relatively recent among the five trials and had a large estimated treatment effect; however, due to its small sample size, the result was judged as showing no significant group difference. Thus, we considered a scenario in which a sufficient sample size is ensured by borrowing participants from the other four trials to detect the treatment effect.

Since the meta-analysis by \cite{sanguankeo2015effects} did not provide the eGFR changes for the treatment and control groups, we made some assumptions. For the control group, the mean change was assumed to be $-\text{(follow up)}/12$, and the variance was assumed to be $n_0 \text(SE)^2/(n_0+n_1)$. For the treatment group, the mean change was assumed to be $-\text{(follow up)}/12 + \text{(total change in eGFR)}$, and the variance was assumed to be $n_1 \text(SE)^2/(n_0+n_1)$. The IPD for the target trial was reconstructed by assuming that the outcomes and covariates follow normal distributions based on the total change in eGFR and the baseline eGFR of the target trial.

The following three sample size settings were then applied:
\begin{enumerate}
    \item \textbf{Target}: The same sample size as the target trial.
    \item \textbf{Target (2:1)}: A 2:1 allocation with 100 participants in the treatment group and 50 participants in the control group.
    \item \textbf{Single-arm}: 100 participants in the treatment group and 0 participants in the control group.
\end{enumerate}
Under these settings, we compared the outcomes estimated using only the target trial with those estimated by borrowing results from the meta-analysis using InMASS. In the Target setting, InMASS was applied to both the treatment and control groups, whereas in the Target (2:1) and Single-arm settings, it was applied only to the control group.

In the meta-analysis of the four trials, the regression parameters $\bbe = (\beta_0, \beta_1, \beta_2)$ were estimated using the model 
\[
\text{Total change in eGFR} = \beta_0 + \beta_1\,(\text{arm}) + \beta_2\,(\text{Baseline eGFR}) + \text{error}.
\]
The results of the meta-analysis for the four trials are shown in Table \ref{tab-rda-meta-res}, and the regression coefficients were estimated as $\hat{\bbe} = (-3.60, 1.42, 0.01)$. Using these results, the total change in eGFR for the four trials was reconstructed based on equation \eqref{eq-ripd}. Although the follow-up period is known to influence the outcome, it was not considered in this study to simplify the discussion.

The results are presented in Table \ref{tab-rda-result}. In the Target setting, the power was insufficient, resulting in a confidence interval that included zero. However, in the Target with InMASS setting, a sufficient number of participants was ensured, leading to a narrower confidence interval. In the Target (2:1) setting, a significant difference was demonstrated, and by using InMASS control, the difference in the outcome of interest could be evaluated with adequate power. In the single-arm setting, comparisons with the control group are challenging. However, by using InMASS control, the estimated outcome difference approached those observed in the Target (2:1) and original Target settings, and sufficiently narrow confidence intervals were obtained. These results indicate that InMASS can supplement an insufficient number of participants in the target trial. This approach achieves higher inferential power compared to relying solely on the target trial, while ensuring that the borrowed participants closely match the target trial setting.

\begin{table}[H]
\centering
\caption{Outcome and covariate of Meta-analysis by \cite{sanguankeo2015effects}.}
\label{tab-rda-meta}
\begin{tabular}{|l|c|c|c|c|cc|}
\hline
             &       &       & Follow up & Total change & \multicolumn{2}{c|}{Baseline eGFR}             \\ \cline{6-7} 
Study        & $n_1$ & $n_0$ & (month)   & in eGFR      & \multicolumn{1}{c|}{Treatment}   & Control     \\ \hline
Yasuda 2004  & 39    & 41    & 12        & -2.0 (0.6)   & \multicolumn{1}{c|}{59.0 (25.6)} & 60.0 (31.2) \\ \hline
Bianchi 2003 & 28    & 28    & 12        & 4.6 (0.2)    & \multicolumn{1}{c|}{50.8 (10.1)} & 50.0 (9.5)  \\ \hline
Rahman 2008  & 779   & 778   & 58        & 0.9 (0.7)    & \multicolumn{1}{c|}{50.8 (8.4)}  & 50.6 (8.2)  \\ \hline
Koren 2009   & 286   & 293   & 54        & 2.1 (0.2)    & \multicolumn{1}{c|}{51.3 (8.5)}  & 51.1 (7.8)  \\ \hline
Sawara 2008  & 22    & 16    & 12        & 4.8 (2.7)    & \multicolumn{1}{c|}{50.7 (16.2)} & 57.3 (18.7) \\ \hline
\end{tabular} \\
{\footnotesize Baseline eGFR: mean (SD), Total change in eGFR: mean (SE).}
\end{table}

\begin{table}[H]
    \centering
    \caption{Meta-regression by data of \cite{sanguankeo2015effects} without target trial.}
    \label{tab-rda-meta-res}
    \begin{tabular}{lccc}
      \hline
     & estimates & SE & $95\%$ CI \\ 
     \hline
        intercept & -3.62 & 16.45 & [-35.85, 28.62] \\ 
        treatment & 1.42 & 2.35 & [-3.19, 6.02] \\ 
        baseline eGFR & 0.01 & 0.31 & [-0.59, 0.62] \\ 
        $\tau^2$ & 10.92 & 9.33 & [3.68, 58.16] \\ 
        \hline
    \end{tabular}
\end{table}

\begin{table}[H]
\centering
\caption{Results of a virtual trial mimicking to target trial setting.}
\label{tab-rda-result}
\begin{tabular}{|l|c|c|c|c|c|c|}
\hline
                                & $n_1$ & $n_0$ & Estimate (SE) & $95\%$ CI         & $t$ value & $Pr(>|t|)$        \\ \hline
Target                          & 22    & 16    & 4.0(2.86)     & {[}-1.83, 9.79{]} & 1.39      & 0.1729            \\ \hline
Target with InMASS              & 22    & 16    & 1.5(0.18)     & {[}1.12, 1.85{]}  & 8.03      & \textless{}0.0001 \\ \hline
Target(2:1)                     & 100   & 50    & 4.8(1.58)     & {[}1.72, 7.95{]}  & 3.07      & 0.0026            \\ \hline
Target(2:1) with InMASS control & 100   & 50    & 6.2(1.99)     & {[}2.27, 10.06{]} & 3.10      & 0.0019            \\ \hline
Single-arm                      & 100   & 0     & NC            & NC                & NC        & NC                \\ \hline
Single-arm with InMASS control  & 100   & 0     & 6.3(1.96)     & {[}2.46, 10.16{]} & 3.21      & 0.0013            \\ \hline
\end{tabular}\\
{\footnotesize CI: confidence interval, IPD: Individual Participant Data, NC: Not Computable, SE: Standard Error.}
\end{table}

\section{Discussion}\label{sec-conclusion}
This study proposes a novel approach, called InMASS, for borrowing participants from AD meta-analysis to ensure sufficient sample sizes in trials with inadequate sample sizes. Unlike existing methods for historical controls—which estimate the CATE for the general or source population—InMASS estimates the CATE specifically for the target trial population. This makes the method particularly useful in scenarios where the population of interest is predefined, such as when conducting a new trial.

The properties demonstrated in this study highlight the utility of the proposed method under certain assumptions and conditions. Proposition \ref{prop1} shows that when the means and variances of covariates are consistently reconstructed from AD, the reconstructed covariates and outcomes obtained from the meta-analytic regression model match the true IPD up to the second moment. This ensures that sufficient IPD for CATE estimation can be reconstructed from the means and variances commonly reported in meta-analyses and research articles. 
Theorem \ref{th1} guarantees the consistency and asymptotic normality of the estimated CATE using the reconstructed outcomes and covariates. These properties depend on an increase in the total number of participants in the source population, rather than on an increase in the number of participants in every individual trial. This is particularly beneficial when increasing the number of participants in the target trial is challenging due to ethical or economic constraints, as it suffices to increase the overall number of participants in the source population. 
Furthermore, Theorem \ref{th2}, valid under the same conditions, suggests that using meta-analysis for CATE estimation yields a smaller variance compared to using only the target trial. When the total number of participants or the number of trials in the source population can be increased, InMASS achieves reduced variance. 
These properties were confirmed through simulation experiments, where InMASS exhibited lower MSE and improved power compared to estimating CATE solely from the target trial. Even in situations where group comparisons are not possible, such as in single-arm trials, InMASS enables such comparisons by borrowing reconstructed IPD from meta-analysis, achieving detection power comparable to that of a 1:1 allocation. In the case study, we demonstrated an example of estimating CATE for a target trial by integrating meta-analysis from a realistic number of trials. Even when the required sample size determined during planning could not be achieved using IPD alone, borrowing reconstructed IPD from meta-analysis allowed the necessary sample size to be secured. Moreover, InMASS enables 1:1 allocation to be achieved using reconstructed IPD, even when the allocation ratios in the target trial are skewed toward the treatment group. This makes the method particularly useful in cases where allocation to the control group is ethically unacceptable.

This study proposes a novel method for efficiently estimating CATE by transferring multiple trials—where only summary statistics are available—to the target trial while utilizing aggregated baseline information under assumptions (A1) and (A2) and conditions (C1) and (C2). 
Conditional independence (A1) requires that at least one trial included in the meta-analysis involves randomization between the group of interest and the control group. This assumption does not require the target trial itself to be randomized, nor does it preclude the inclusion of observational studies in the meta-analysis. Thus, theoretically, (A1) is satisfied if the meta-analysis includes one randomized controlled trial and a sufficiently large observational study. The allowance for such a relaxed conditional independence (A1) is due to the covariate shift assumption (A2), which assumes that the conditional distribution of outcomes given covariates is the same in the target and source trials. 
Although it is challenging to verify the validity of the assumption (A2) statistically, differences between trials can be accounted for using random effects in the meta-analysis, allowing slight variations in the conditional distribution across trials to be absorbed as random effects. Condition (C1) requires the consistency of the means and variances used in the meta-analysis. Although (C1) is a common assumption in AD meta-analysis, it may not hold for non-collapsible outcomes such as hazard ratios or odds ratios; in such cases, the outcome can be transformed into a collapsible quantity to make the method applicable. In this study, the target of estimation, CATE, is a first moment, and the summary data used in the meta-analysis (means and variances) are sufficient for the analysis. However, if the target of estimation includes medians, survival curves, or entire distributions, information beyond means and variances would be required from the source trials. The assumption of covariate shift also justifies the use of weighted regression based on density ratios, although it may be possible to relax this assumption further using more advanced theories in transfer learning. Condition (C2) requires that the entire source population adequately covers the support of the target trial, ensuring the stability and finite variance of InMASS. Since the target trial itself is included in the source population, this condition is generally satisfied; however, it may not hold if the covariate distributions of the source and target trials are extremely different, in which case estimation should be performed using only the target trial without leveraging meta-analysis.

InMASS is particularly useful when the target population of interest is predefined and a meta-analysis has already been conducted, as it enables the estimation of the CATE for the target population by integrating the results of the meta-analysis with those of an additional small-sample trial.

\section*{Acknowledgements}
This work was supported by JSPS KAKENHI (Grant numbers JP24K23862).
\vspace*{-8pt}

\section*{Supplementary Material}
The R code used for the simulation and case study, together with the function required to run InMASS, is publicly available at: \url{https://github.com/keisuke-hanada/InMASS}.
\vspace*{-8pt}

\appendix
\section{Appendix}\label{app}

\subsection{Proofs}

\begin{proof}[Proof of Proposition \ref{prop1}]
First, suppose that $F_{\X}$ is known and that $X_{ki}^*$, for $i=1,\dots,n_k$, are independently sampled from $F_{\X}$. In this case, 
\[
{\rm E}[X_{ki}^* \mid \hat{\bbe}_{\cM}] = \bmu_k \quad \text{and} \quad {\rm V}[X_{ki}^* \mid \hat{\bbe}_{\cM}] = \bsi_{\bmu_k}^2 \I_p.
\]
Therefore, the mean and variance of $\hat{Y}_{ki}^*$ are given by
\begin{align*}
    {\rm E}[\hat{Y}_{ki}^* \mid \hat{\bbe}_{\cM}] &= {\rm E}[X_{ki}^* \hat{\bbe}_{\cM} + \epsilon_{ki}^* \mid \hat{\bbe}_{\cM}] = \bmu_k \hat{\bbe}_{\cM}, \\
    {\rm V}[\hat{Y}_{ki}^* \mid \hat{\bbe}_{\cM}] &= {\rm V}[X_{ki}^* \hat{\bbe}_{\cM} + \epsilon_{ki}^* \mid \hat{\bbe}_{\cM}] = \sigma_k^2.
\end{align*}

Next, suppose that $F_{\X}$ is unknown; however, assume that condition (C1) holds. In this case, if $X_{ki}^*$, for $i=1,\dots,n_k$, are independently sampled from $\hat{F}_{\X}(\cdot \mid \bar{\x}_k, \hat{\bsi}_{\bar{\x}_k}^2)$, then 
\[
{\rm E}[X_{ki}^* \mid \hat{\bbe}_{\cM}] = \bar{\x}_k \quad \text{and} \quad {\rm V}[X_{ki}^* \mid \hat{\bbe}_{\cM}] = \hat{\bsi}_{\bar{\x}_k}^2 \I_p.
\]
Therefore, the mean and variance of $\hat{Y}_{ki}^*$ are given by
\begin{align*}
    {\rm E}[\hat{Y}_{ki}^* \mid \hat{\bbe}_{\cM}] &= \bar{\x}_k \hat{\bbe}_{\cM} \to \bmu_k \hat{\bbe}_{\cM}, \\
    {\rm V}[\hat{Y}_{ki}^* \mid \hat{\bbe}_{\cM}] &= \hat{\sigma}_k^2 \to \sigma_k^2,
\end{align*}
as $n_k \to \infty$ by condition (C1).
\end{proof}

\begin{lemma}
\label{lemma-exchangability}
    Under the assumptions (A1) and (A2) and the condition (C2), 
    \[
        {\rm E}_{T}[Y_1 - Y_0] = {\rm E}_{T}[Y \mid Z=1] - {\rm E}_{T}[Y \mid Z=0]
    \]
    for the target population $T$.
\end{lemma}
\begin{proof}[Proof of Lemma \ref{lemma-exchangability}]
    From (A2), for any $A, B \in \cS = \{T, \cM\}$ with $A \ne B$, we have 
    \[
    P_A(Y \mid X, Z) = P_B(Y \mid X, Z).
    \]
    Additionally, by (C2), for any $S \in \cS$, 
    \[
    \sup_{i \in S} \left| \frac{dP_T(x_{Si}, z_{Si})}{dP_{\cS}(x_{Si}, z_{Si})} \right| < \infty.
    \]
    From (A1), there exists some $S_0 \in \cS$ such that 
    \[
    0 < P(Z = 1 \mid S_0) < 1 \quad \text{and} \quad \{Y_1, Y_0\} \indep Z \mid S_0.
    \]
    Thus, the following calculations can be performed:
    \begin{align*}
        {\rm E}_{T}[Y_1 - Y_0] 
        &= {\rm E}_{\cS}\left[ \frac{dP_{T}(X,Z)}{dP_{\cS}(X,Z)} (Y_1 - Y_0) \right] \\
        &= {\rm E}_{\cS}\left[ {\rm E}_{\cS}\left[\frac{dP_{T}(X,Z)}{dP_{\cS}(X,Z)} (Y_1 - Y_0) \mid X, Z \right] \right] \\
        &= {\rm E}_{\cS}\left[ {\rm E}_{S_0}\left[ \frac{dP_{T}(X,Z)}{dP_{\cS}(X,Z)} (Y_1 - Y_0) \mid X, Z \right] \right] \quad (\text{from (A2)}) \\
        &= {\rm E}_{\cS}\left[ \frac{dP_{T}(X,1)}{dP_{\cS}(X,1)} \, {\rm E}_{S_0}[Y \mid X, Z=1] \right] - {\rm E}_{\cS}\left[ \frac{dP_{T}(X,0)}{dP_{\cS}(X,0)} \, {\rm E}_{S_0}[Y \mid X, Z=0] \right] \quad (\text{from (A1)}) \\
        &= {\rm E}_{\cS}\left[ \frac{dP_{T}(X,1)}{dP_{\cS}(X,1)} \, {\rm E}_{\cS}[Y \mid X, Z=1] \right] - {\rm E}_{\cS}\left[ \frac{dP_{T}(X,0)}{dP_{\cS}(X,0)} \, {\rm E}_{\cS}[Y \mid X, Z=0] \right] \quad (\text{from (A2)}) \\
        &= \int_{\cS} \frac{dP_{T}(x,1)}{dP_{\cS}(x,1)} \, y \, dP_{\cS}(x,y,1) - \int_{\cS} \frac{dP_{T}(x,0)}{dP_{\cS}(x,0)} \, y \, dP_{\cS}(x,y,0) \\
        &= {\rm E}_{T}[Y \mid Z=1] - {\rm E}_{T}[Y \mid Z=0].
    \end{align*}
\end{proof}

\begin{proof}[Proof of Theorem \ref{th1}]
Let $\bar{Y}_j$ be the weighted mean of $\hat{Y}_{Sj}$, defined as  
\[
\bar{Y}_j = \frac{\sum_{S \in \cS} \sum_{i=1}^{n_S} w_{Si}\,\mathbbm{1}(z_{Si}=j)\,\hat{Y}_{Si}}{\tilde{N}_j}.
\]
From (A1), $\bar{Y}_j$ is well-defined for $j=0,1$, and $\hat{\delta}_R = \bar{Y}_1 - \bar{Y}_0$. The quantity $\bar{Y}_j$ can be calculated as follows:
\begin{align*}
    \bar{Y}_j &= \frac{N_j^{-1}\sum_{S \in \cS} \sum_{i=1}^{n_S} w_{Si}\,\mathbbm{1}(z_{Si}=j)\,\hat{Y}_{Si}}{N_j^{-1}\tilde{N}_j} \\
    &\to \frac{{\rm E}_{\cS}\bigl[ w(X,Z)Y \mid Z=j \bigr]}{{\rm E}_{\cS}\bigl[ w(X,Z) \mid Z=j \bigr]} \quad (\text{as } N_j \to \infty \text{ and (A1)}) \\
    &= \frac{{\rm E}_{T}\bigl[ Y \mid Z=j \bigr]}{{\rm E}_{T}\bigl[1 \mid Z=j \bigr]}
    = {\rm E}_{T}\bigl[ Y \mid Z=j \bigr].
\end{align*}
Thus, by Lemma \ref{lemma-exchangability}, 
\[
\hat{\delta}_R = \bar{Y}_1 - \bar{Y}_0 \to {\rm E}_{T}\bigl[ Y \mid Z=1 \bigr] - {\rm E}_{T}\bigl[ Y \mid Z=0 \bigr] = {\rm E}_{T}\bigl[ Y_1 - Y_0 \bigr] = \delta_T
\]
as $N_0, N_1 \to \infty$. Since $\hat{\sigma}_{Rj}^2$ is obtained by replacing ${\rm E}_{T}\bigl[ Y \mid Z=j \bigr]$ in $\sigma_{Rj}^2$ with its estimator $\bar{Y}_j$, the continuous mapping theorem implies that $\bar{Y}_j \to {\rm E}_{T}\bigl[ Y \mid Z=j \bigr]$ as $N_j \to \infty$, which in turn yields $\hat{\sigma}_{Rj}^2 \to \sigma_{Rj}^2$. 

Furthermore, by (C2), there exists some 
\[
w_M = \sup_{i \in S,\, S \in \cS} \left| \frac{dP_T(x_{Si}, z_{Si})}{dP_{\cS}(x_{Si}, z_{Si})} \right| < \infty,
\]
and
\begin{align*}
    \frac{1}{N_j} \sigma_{Rj}^2 &\le \frac{N_j \, w_M^2 \max_{S \in \cS} {\rm V}_{\cS}[\hat{Y}_S \mid Z=j]}{N_j^2 \, \bar{w}_j^2} = O\bigl(N_j^{-1}\bigr),
\end{align*}
where $\bar{w}_j = N_j^{-1} \sum_{S \in \cS} \sum_{i=1}^{n_S} w_{Si}\,\mathbbm{1}(z_{Si}=j)$. Therefore,
\[
{\rm V}[\hat{\delta}_R] = \frac{1}{N_1}\sigma_{R1}^2 + \frac{1}{N_0}\sigma_{R0}^2 = O\bigl(N_1^{-1} + N_0^{-1}\bigr),
\]
satisfying the Lindeberg condition. Thus,
\[
\frac{\hat{\delta}_R - \delta_T}{\sqrt{\tilde{N}_1^{-1}\hat{\sigma}_{R1}^2 + \tilde{N}_0^{-1}\hat{\sigma}_{R0}^2}} \to_d N(0,1)
\]
as $N_j \to \infty$, where $\to_d$ denotes convergence in distribution.
\end{proof}

\begin{proof}[Proof of Theorem \ref{th2}]
The ratio of ${\rm V}[\hat{\delta}_R]$ to ${\rm V}[\hat{\delta}_T]$ is given by  
\[
\frac{{\rm V}[\hat{\delta}_R]}{{\rm V}[\hat{\delta}_T]} = \frac{\frac{\sigma_{R1}^2}{\tilde{N}_1} + \frac{\sigma_{R0}^2}{\tilde{N}_0}}{\frac{\sigma_{T1}^2}{n_{T1}} + \frac{\sigma_{T0}^2}{n_{T0}}},
\]
which can be bounded as follows:
\begin{align*}
    \frac{{\rm V}[\hat{\delta}_R]}{{\rm V}[\hat{\delta}_T]} 
    &\le \frac{\max\left( \frac{1}{\tilde{N}_1}, \frac{1}{\tilde{N}_0} \right)(\sigma_{R1}^2+\sigma_{R0}^2)}{\min\left( \frac{1}{n_{T1}}, \frac{1}{n_{T0}} \right)(\sigma_{T1}^2+\sigma_{T0}^2)} \\
    &\le C\, \frac{\max\left( n_{T1}, n_{T0} \right)}{\min\left( \tilde{N}_1, \tilde{N}_0 \right)}
    = O\left( \frac{n_T}{\tilde{N}} \right),
\end{align*}
where 
\[
C = \frac{\sigma_{R1}^2+\sigma_{R0}^2}{\sigma_{T1}^2+\sigma_{T0}^2}.
\]
Moreover, since 
\[
\tilde{N}_j \ge K\, \min\{n_{kj}\}_{k=1,\dots, K,\, j=0,1},
\]
we obtain  
\begin{align*}
    \frac{{\rm V}[\hat{\delta}_R]}{{\rm V}[\hat{\delta}_T]} 
    &\le C\, \frac{\max\left( n_{T1}, n_{T0} \right)}{K\, \min\{n_{kj}\}_{k=1,\dots, K,\, j=0,1}}
    = O\left( K^{-1} \right).
\end{align*}
When $w_{S_k i} = 0$ for all $k=1,\dots, K$ and $w_{Ti}=1$, we have $\tilde{N} = \tilde{N}_1 + \tilde{N}_0 = n_T$ and $\tilde{n}_k = 0$. In this case, 
\[
\frac{{\rm V}[\hat{\delta}_{R}]}{{\rm V}[\hat{\delta}_{T}]} = 1.
\]
\end{proof}

\begin{proof}[Proof of Theorem \ref{th-central-limit-theorem}]
The estimator $\hat{\bbe}_T$ can be written as 
\[
\hat{\bbe}_T = \bbe_T + (\X_{\cS}^T \W \X_{\cS})^{-1} \X_{\cS}^T \W \bep.
\]
From the model \eqref{eq-reg-model}, $\X_{\cS}$ and $\bep_{\cS}$ are independent. The expectation of $\hat{\bbe}_T$ is computed as
\begin{align*}
    {\rm E}_{\cS}[\hat{\bbe}_T] &= \bbe_T + {\rm E}_{\cS}\Bigl[(\X_{\cS}^T \W \X_{\cS})^{-1} \X_{\cS}^T \W \bep_{\cS}\Bigr] \\
    &= \bbe_T + {\rm E}_{\cS}\Bigl[(\X_{\cS}^T \W \X_{\cS})^{-1} \X_{\cS}^T \W\Bigr]\,{\rm E}_{\cS}[\bep_{\cS}] = \bbe_T.
\end{align*}
From condition (C3), 
\[
N \, (\X_{\cS}^T \W \X_{\cS})^{-1} \to \Q^{-1} 
\]
in probability as $N \to \infty$. Applying the central limit theorem (Corollary 2.7.1 of \cite{lehmann1999elements}) to $N^{-1/2} \X_{\cS}^T \W \bep_{\cS}$ and Slutsky's theorem (Theorem 2.3.3 of \cite{lehmann1999elements}), we obtain
\[
\sqrt{N} (\hat{\bbe}_T - \bbe_T) \to N\bigl(\mathbf{0}, \bSi_{\hat{\bbe}_T}\bigr)
\]
in distribution as $N \to \infty$.

Finally, we show that $\hat{\bSi}_{\hat{\bbe}_T} \to \bSi_{\hat{\bbe}_T}$. Let $\e = \hat{\Y}_{\cS} - \X_{\cS}^T \hat{\bbe}_T$ be the residual of the model \eqref{eq-reg-model}; note that the residual can be expressed as $\e = \bep_{\cS} + \X_{\cS} (\hat{\bbe}_T - \bbe_T)$. The outer product of the weighted residual is
\begin{align*}
    \e \e^T &= \bep_{\cS} \bep_{\cS}^T + \bep_{\cS} (\hat{\bbe}_T - \bbe_T)^T \X_{\cS}^T + \X_{\cS} (\hat{\bbe}_T - \bbe_T) \bep_{\cS}^T \\
    &\quad + \X_{\cS} (\hat{\bbe}_T - \bbe_T)(\hat{\bbe}_T - \bbe_T)^T \X_{\cS}^T \\
    &= \bep_{\cS} \bep_{\cS}^T + O_p(N^{-1/2}).
\end{align*}
From the covariate shift assumption (A2), we have 
\[
{\rm E}_{\cS} [w_i w_j \epsilon_i \epsilon_j] = {\rm E}_{T} [\epsilon_i \epsilon_j],
\]
where $w_i$ is the $i$-th element of $\W$ and $\epsilon_i$ is the $i$-th element of $\bep_{\cS}$.
Let $X_i$ is a $(p+2) \times 1$ vector which is the $i$-th row of the matrix $\X_{\cS}$, then
\begin{align*}
    N\, \hat{\bSi}_{\hat{\bbe}_T} 
    &= \{ \Q + o_p(1) \}^{-1} \frac{1}{N} \X_{\cS}^T \W \bSi_{\hat{\bbe}_T} \W \X_{\cS} \{ \Q + o_p(1) \}^{-1} \\
    &= \{ \Q + o_p(1) \}^{-1} \frac{1}{N} \diag\Bigl(\X_{\cS}^T \W \Bigl\{ \W \bep_{\cS} \bep_{\cS}^T \W + O_p(N^{-1/2})\Bigr\} \W \X_{\cS}\Bigr) \{ \Q + o_p(1) \}^{-1} \\
    &= \{ \Q + o_p(1) \}^{-1} \frac{1}{N} \sum_{i=1}^N X_i X_i^T\, w_i^4\, (\epsilon_i^2 + O_p(N^{-1/2})) \{ \Q + o_p(1) \}^{-1} \\
    &\to \Q^{-1} {\rm E}_{\cS} [X_i X_i^T w_i^4 \epsilon_i^2] \Q^{-1} \\
    &= \Q^{-1} {\rm E}_{T} [X_i X_i^T w_i^2 \epsilon_i^2] \Q^{-1} \\
    &= \Q^{-1} {\rm E}_{T} [\X_{\cS}^T \W \bep_{\cS} \bep_{\cS}^T \W \X_{\cS}] \Q^{-1} \\
    &= \Q^{-1} \X_{\cS}^T \W \bSi_{\bep} \W \X_{\cS} \Q^{-1},
\end{align*}
as $N \to \infty$. Thus, we have $\hat{\bSi}_{\hat{\bbe}_T} \to \bSi_{\hat{\bbe}_T}$.
\end{proof}

\section{Supplemental simulation}\label{supp-simulation}
\begin{figure}[H]
    \centering
    \includegraphics[width=0.9\linewidth]{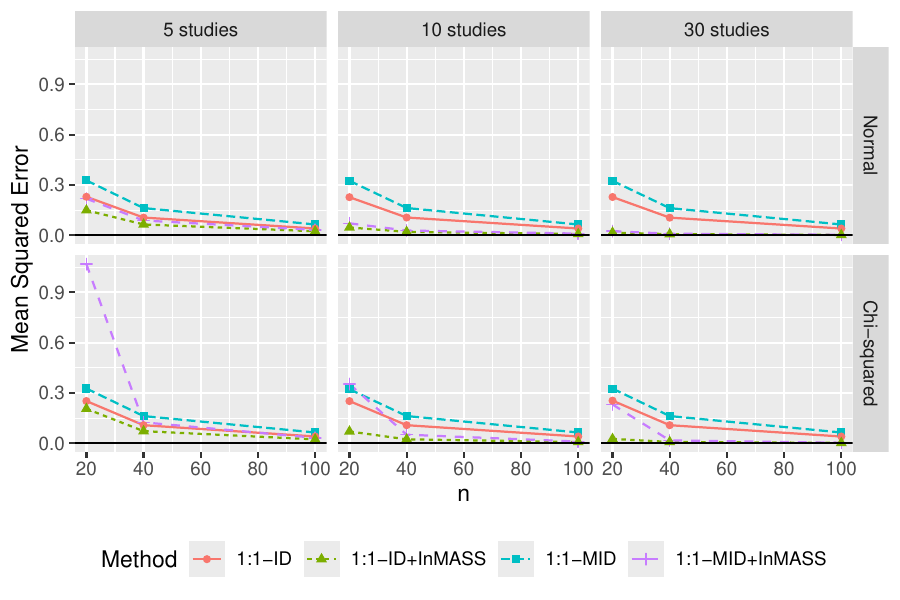}
    \caption{Means squared error of CATE for target population with equal randomization. ID: model identified; MID: model misidentified.}
    \label{fig-1to1-mse}
\end{figure}
\begin{figure}[H]
    \centering
    \includegraphics[width=0.9\linewidth]{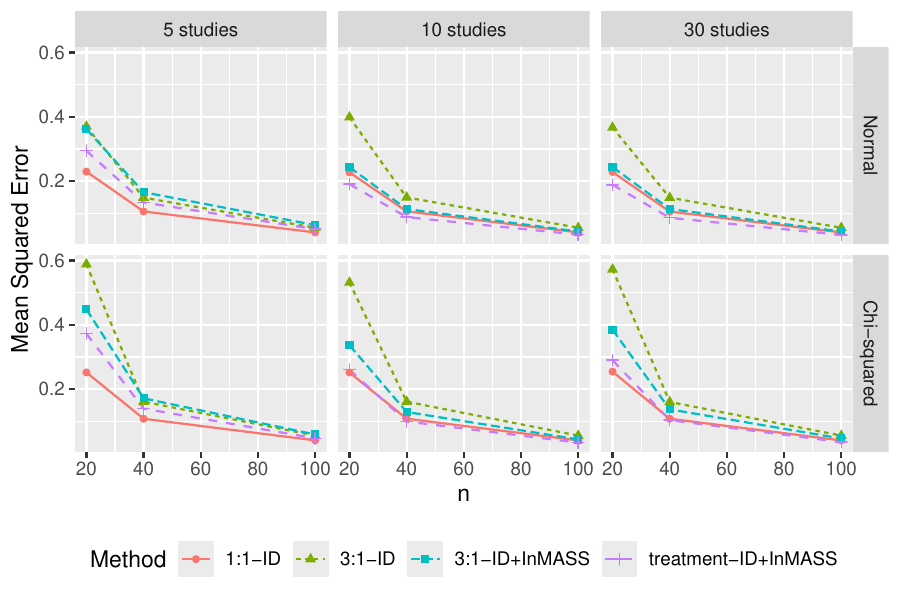}
    \caption{Means squared error of CATE for target population with unequal randomization. The model specified case. ID: model identified; MID: model misidentified.}
    \label{fig-3to1-mse}
\end{figure}
\begin{figure}[H]
    \centering
    \includegraphics[width=0.9\linewidth]{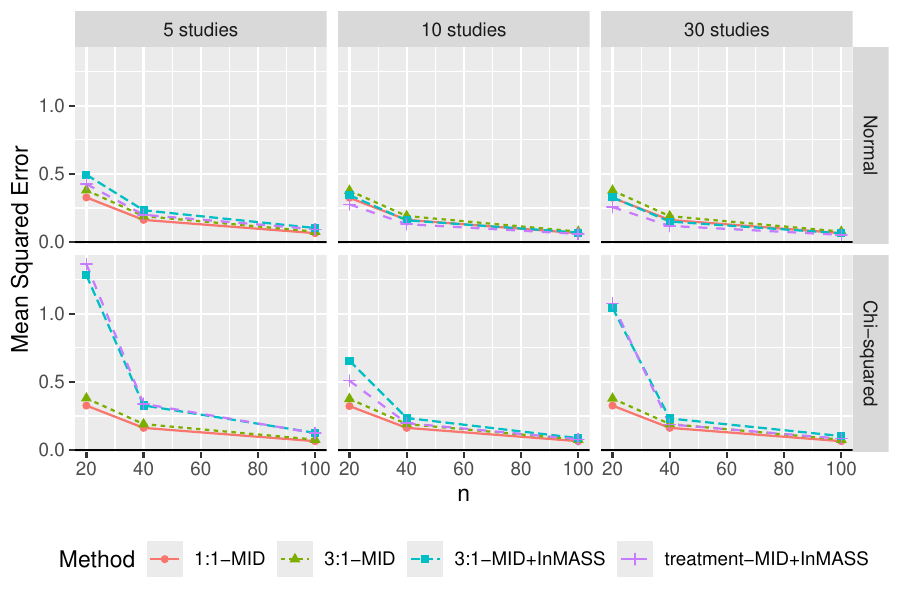}
    \caption{Means squared error of CATE for target population. The model misspeficied case. ID: model identified; MID: model misidentified.}
    \label{fig-controlAD-mse}
\end{figure}

\begin{figure}[H]
    \centering
    \includegraphics[width=1\linewidth]{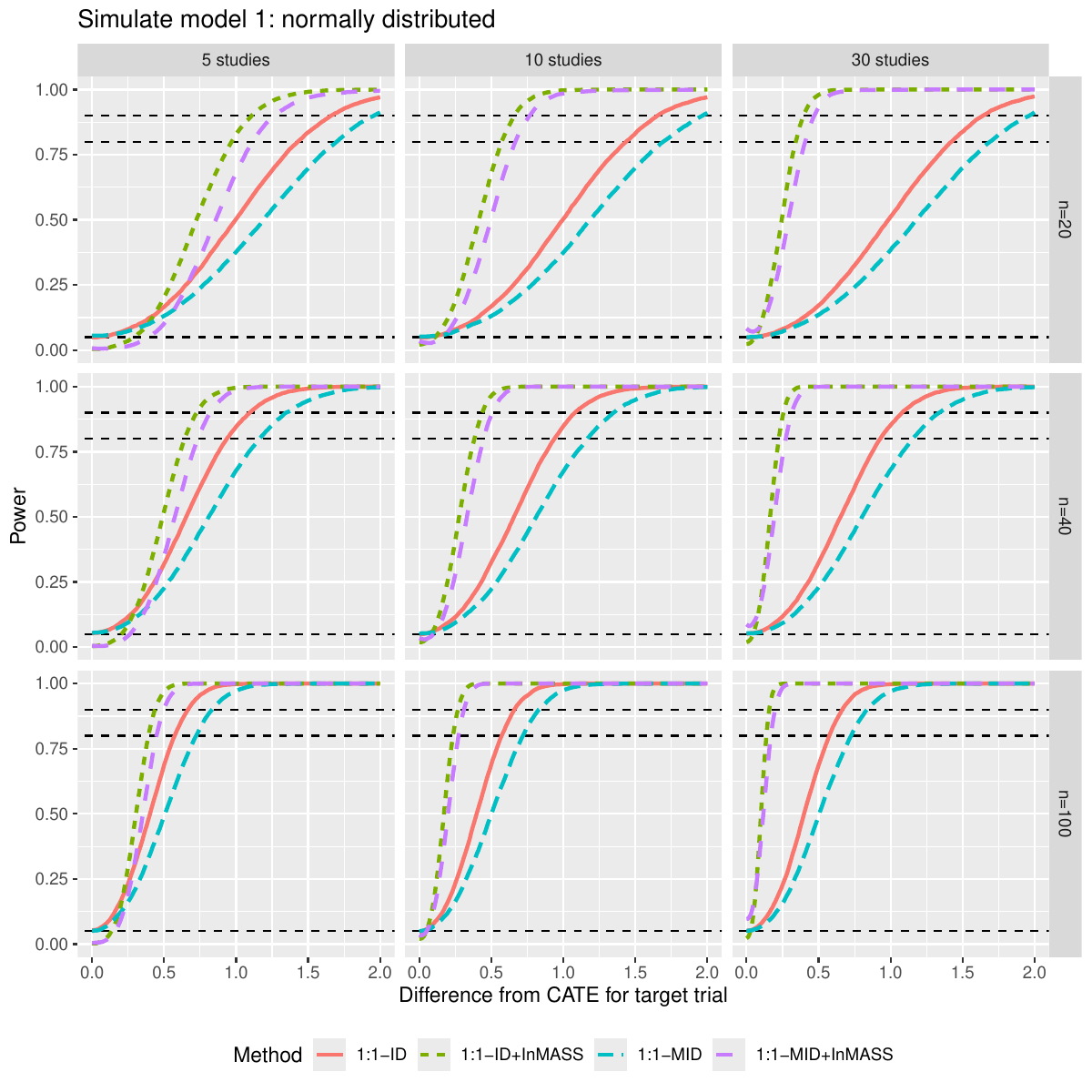}
    \caption{Power plot for the simulation model \eqref{eq-sim-model} with equal randomization and normally distributed covariates.}
    \label{fig-1to1-power1}
\end{figure}
\begin{figure}[H]
    \centering
    \includegraphics[width=1\linewidth]{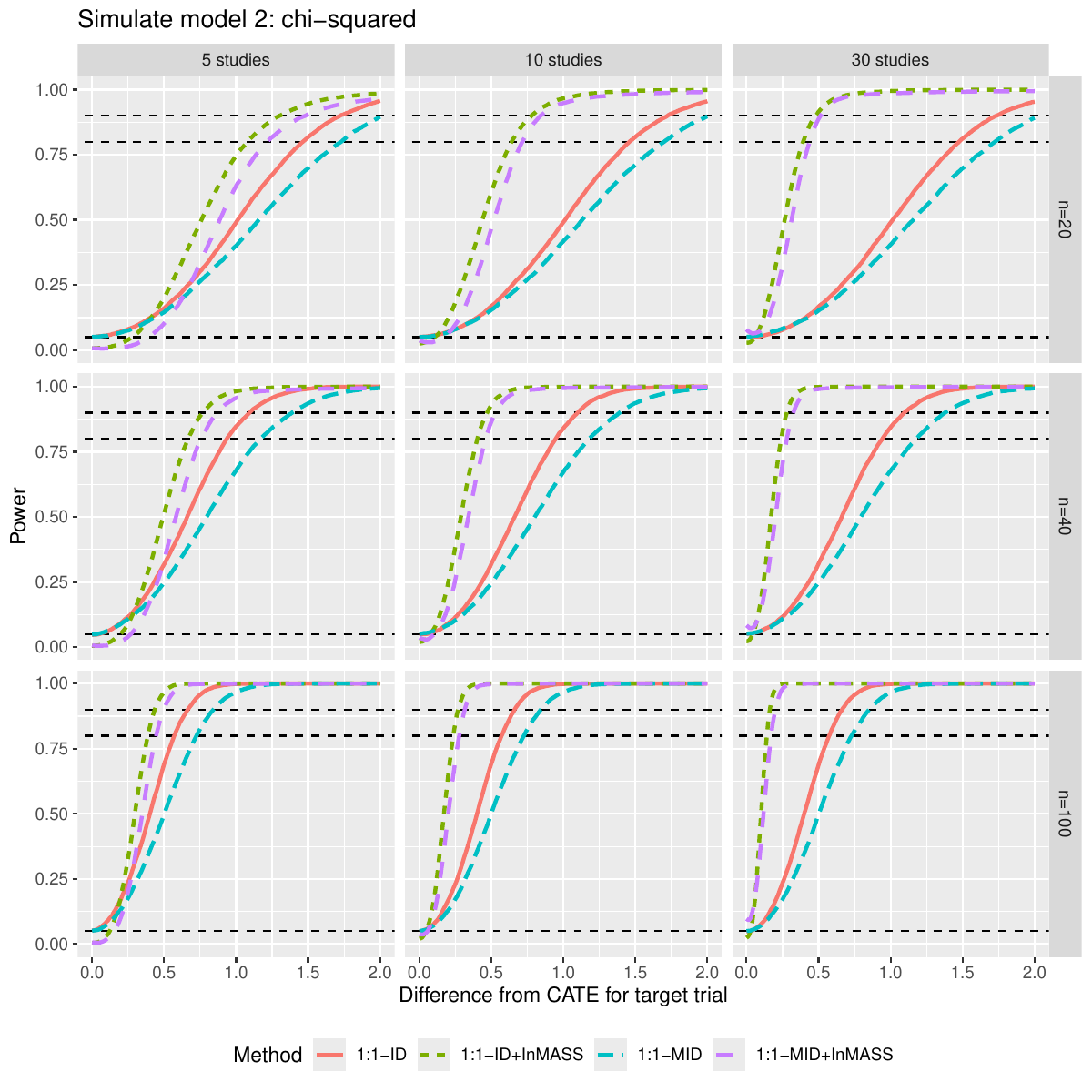}
    \caption{Power plot for the simulation model \eqref{eq-sim-model} with equal randomization and covariates following a chi-squared distribution.}
    \label{fig-1to1-power2}
\end{figure}

\begin{figure}[H]
    \centering
    \includegraphics[width=1\linewidth]{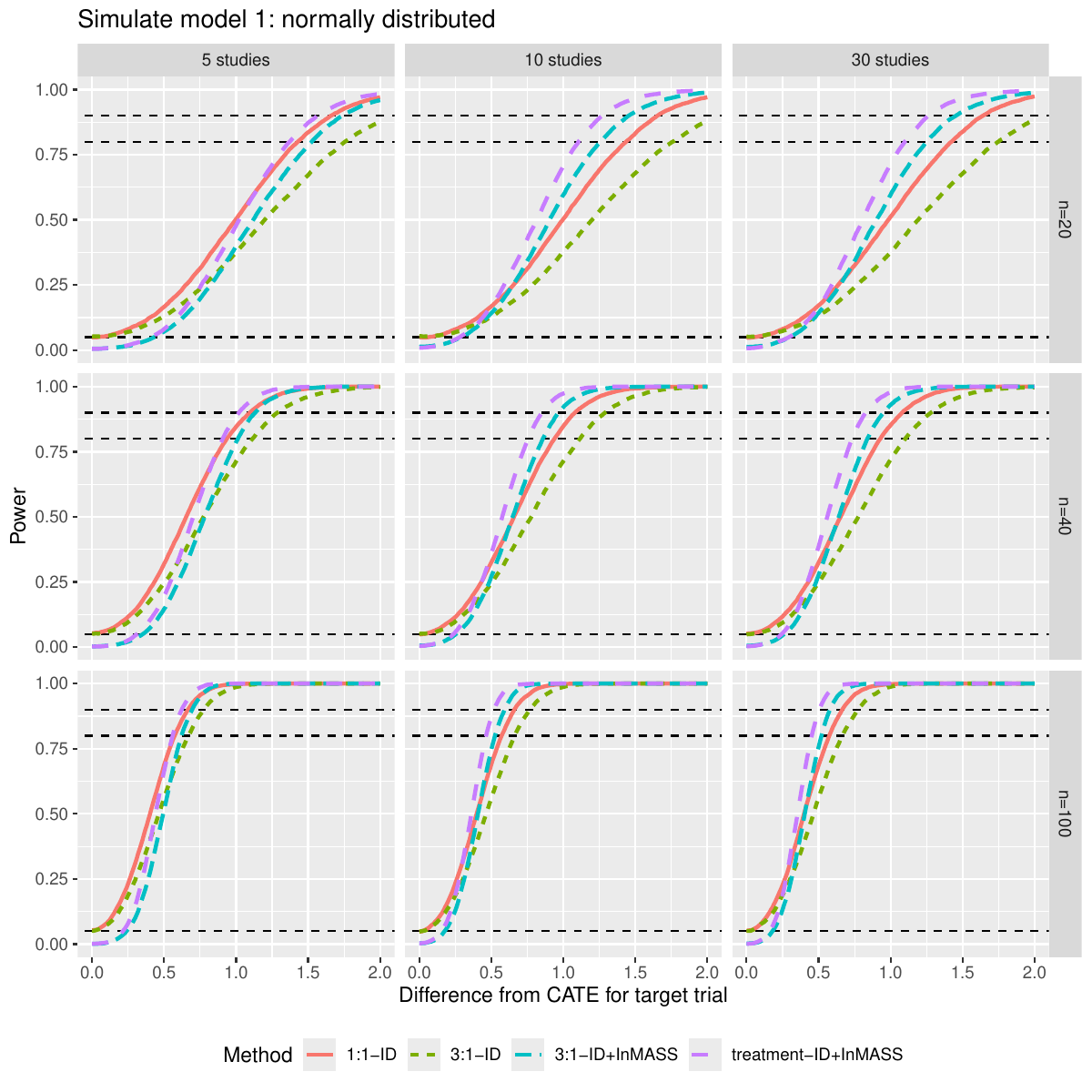}
    \caption{Power plot for the simulation model \eqref{eq-sim-model} with unequal randomization and normally distributed covariates.}
    \label{fig-controlAD-id-power1}
\end{figure}

\begin{figure}[H]
    \centering
    \includegraphics[width=1\linewidth]{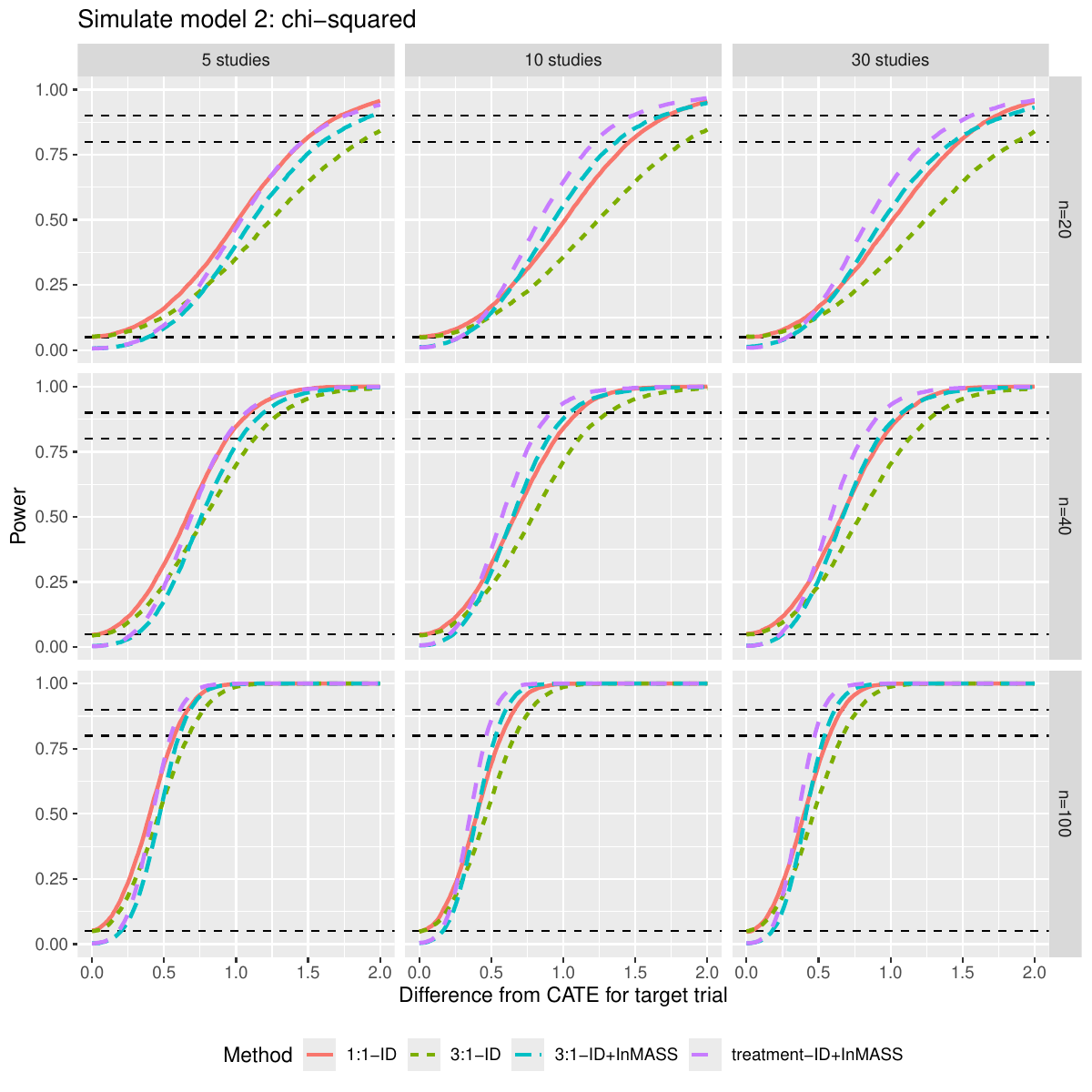}
    \caption{Power plot for the simulation model \eqref{eq-sim-model} with unequal randomization and covariates following a chi-squared distribution.}
    \label{fig-controlAD-id-power2}
\end{figure}

\begin{figure}[H]
    \centering
    \includegraphics[width=1\linewidth]{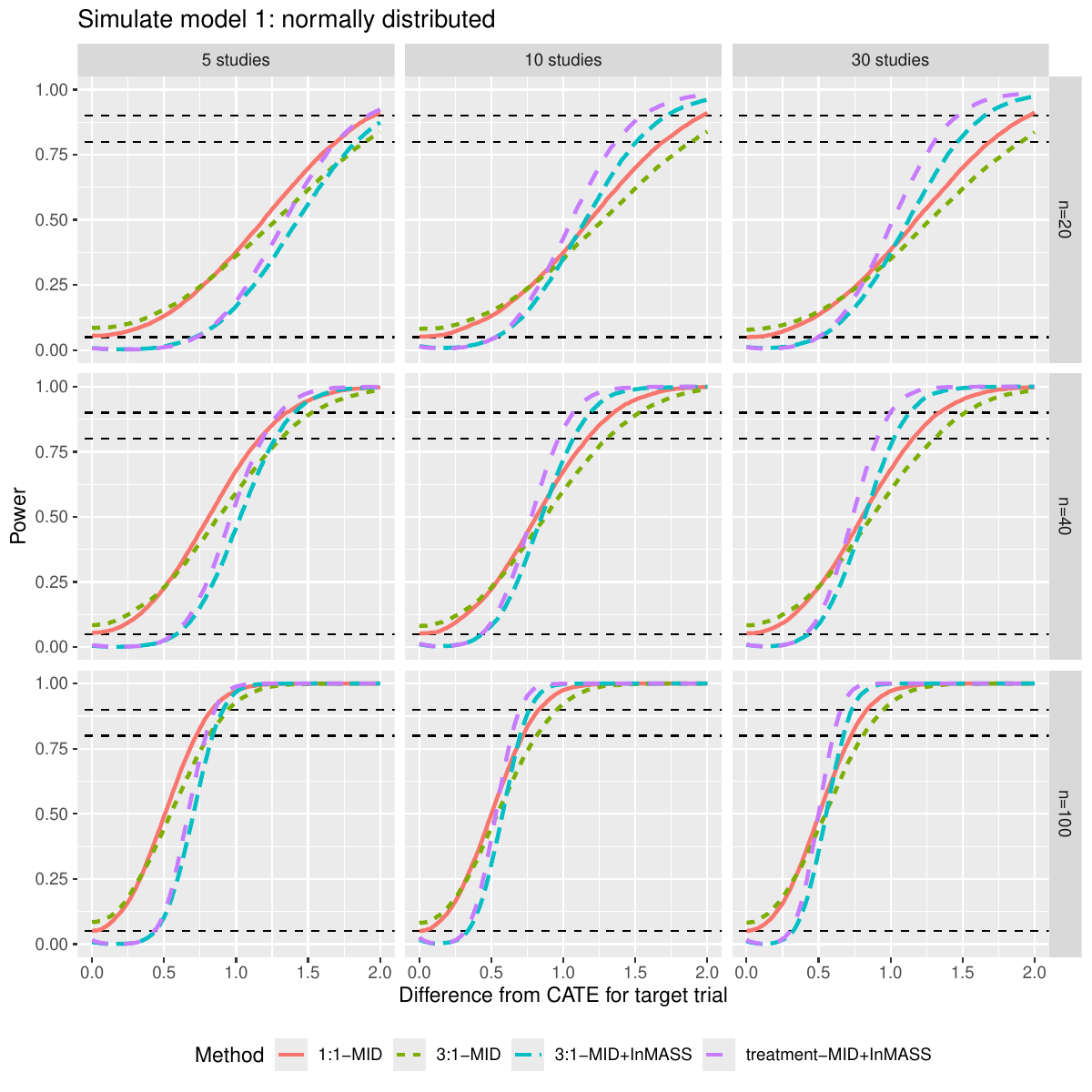}
    \caption{Power plot for the simulation model \eqref{eq-sim-model} with unequal randomization, normally distributed covariates, and model misspecification.}
    \label{fig-controlAD-mid-power1}
\end{figure}

\begin{figure}[H]
    \centering
    \includegraphics[width=1\linewidth]{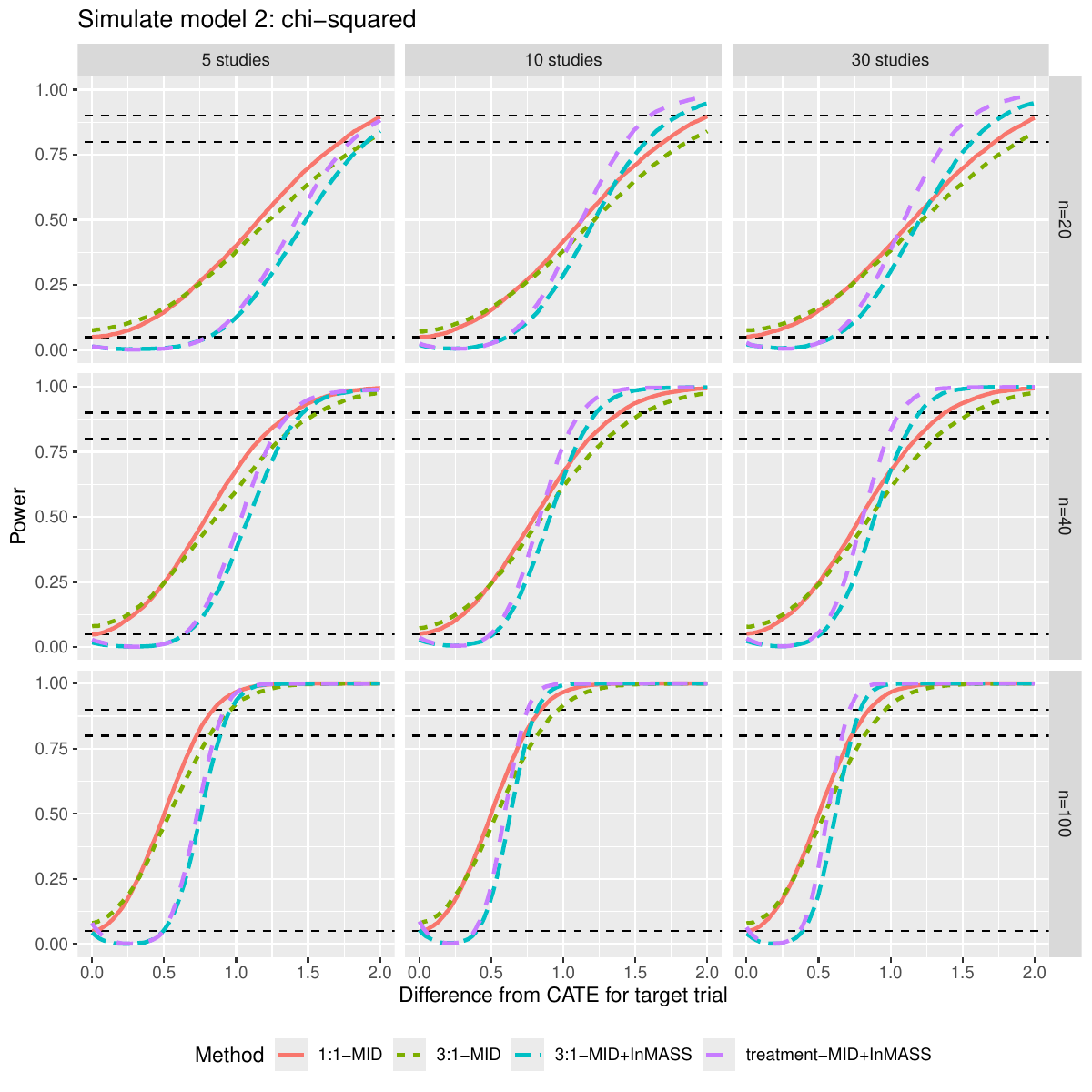}
    \caption{Power plot for the simulation model \eqref{eq-sim-model} with unequal randomization, covariates following a chi-squared distribution, and model misspecification.}
    \label{fig-controlAD-mid-power2}
\end{figure}

\bibliography{main.bib} 
\bibliographystyle{abbrvnat}

\end{document}